\documentclass[10pt, conference, letterpaper]{IEEEtran}
\IEEEoverridecommandlockouts
\usepackage{cite}
\usepackage{amsmath,amssymb,amsfonts}
\usepackage{booktabs} 
\usepackage{graphicx}
\usepackage{textcomp}
\usepackage{algorithm}
\usepackage{algpseudocode}
\usepackage{xspace}
\usepackage{hyperref}
\usepackage{amsfonts}
\usepackage{amsthm}
\usepackage{balance}
\usepackage{subcaption}

\newtheorem{theorem}{Theorem}

\newtheorem{observation}[theorem]{Observation}

\def\BibTeX{{\rm B\kern-.05em{\sc i\kern-.025em b}\kern-.08em
    T\kern-.1667em\lower.7ex\hbox{E}\kern-.125emX}}

\newcommand{\update}{R\xspace}
\newcommand{\centered}[1]{\begin{tabular}{c} #1 \end{tabular}}

\newcommand{\PLH}{\text{x}}
\newcommand{\win}{W}
\newcommand{\greedy}{greedy $k$-matching\xspace}
\newcommand{\bma}{\textsc{Online-BMA}\xspace}
\def\card#1{\lvert #1 \rvert}
\DeclareMathOperator{\outdeg}{out-deg}
\DeclareMathOperator{\indeg}{in-deg}
\DeclareMathOperator{\Fbfs}{ForwardBFS}
\DeclareMathOperator{\Bbfs}{BackwardBFS}
\DeclareMathOperator{\dist}{dist}
\newcommand{\model}{\textsc{TMT}\xspace}
\newcommand{\netw}{N}
\DeclareMathOperator{\apl}{APL}
 \newcommand{\algo}{\textsc{GreedyEgoTrees}\xspace}
\newcommand{\D}{\mathcal{D}}
\def\p{\mathrm{p}}
\def\nets{\mathcal{N}}
\DeclareMathOperator*{\argmin}{\arg\min}
\DeclareMathOperator{\dan}{DAN}

\usepackage{tikz}
\usetikzlibrary{shapes.geometric, arrows}
\usetikzlibrary{fit,calc}
\newcommand*{\tikzmk}[1]{\tikz[remember picture,overlay,] \node (#1) {};\ignorespaces}
\newcommand{\boxit}[5]{\tikz[remember picture,overlay]{\node[yshift=3pt,fill=#1,opacity=.25,fit={($(A)-(#2\linewidth,#3\baselineskip)$)($(B)+(#4\linewidth,#5\baselineskip)$)}] {};}\ignorespaces}
\colorlet{mode1}{red!40}
\colorlet{mode2}{cyan!60}
\colorlet{mode3}{green!70}
\colorlet{mode4}{gray!70}

\begin{document}

\title{One For All, All for One! \\ Efficient Self-Adjusting Leaf-Spine Network  \thanks{Don't forget to thanks ERC.} }
\title{A Perfect Match? \\ Exploring (imperfect) Self-Adjusting Leaf-Spine Network Designs  \thanks{Don't forget to thanks ERC.}}
\title{Efficient Self-Adjusting Optical Leaf-Spine Network Design}
\title{Efficient Self-Adjusting Optical Leaf-Spine Network}
\title{Self-Adjusting Ego-Trees Topology for Reconfigurable Datacenter Networks}

\author{
\IEEEauthorblockN{Chen Griner, Gil Einziger, and Chen Avin}
\IEEEauthorblockA{\textit{Ben-Gurion University of the Negev}, Israel}
} 


\maketitle
\pagestyle{plain}

\begin{abstract}
State-of-the-art topologies for datacenters (DC) and high-performance computing (HPC) networks are demand-oblivious and static. Therefore, such network topologies are optimized for the worst-case traffic scenarios and can't take advantage of changing demand patterns when such exist. However, recent optical switching technologies enable the concept of dynamically reconfiguring circuit-switched topologies in real-time. This capability opens the door for the design of self-adjusting networks: networks with demand-aware and dynamic topologies in which links between nodes can be established and re-adjusted online and respond to evolving traffic patterns. 

This paper studies a recently proposed model for optical leaf-spine reconfigurable networks. 
We present a novel algorithm, \algo,  that dynamically changes the network topology. The algorithm greedily builds ego trees for nodes in the network, where nodes cooperate to help each other, taking into account the global needs of the network. We show that \algo has nice theoretical properties, outperforms other possible algorithms (like static expander and greedy dynamic matching) and can significantly improve the average path length for real DC and HPC traces. 

\end{abstract}



\section{Introduction}
Communication networks in general and datacenter (DC) networks, in particular, have become a critical infrastructure in our digital society. The popularity of data-centric applications,  e.g., related to entertainment, science, social networking, and business, is rapidly increasing. The ongoing COVID-19 pandemic has further highlighted the need for an efficient
communication infrastructure, which is now critical for, e.g., online teaching, virtual conferences, and health \cite{OECD20Keeping}.  

Network topology is directly related to network performance in terms of delay, throughput, and reliability. Therefore,  research and innovations in network topologies are a fundamental part of network design in industry and academia alike~\cite{rotornet,ballani2020sirius,flexspander,xpander,clos}. 
State-of-the-art (SoA) datacenter network designs typically rely on \emph{static} and \emph{demand-oblivious} optical switches. However, 
DC  networks currently serve a variety of specific applications such as web search, machine learning (ML), High-performance-computing (HPC), or distributed storage.
Each 
application creates different and dynamic demand patterns
and the overall traffic may be a mix of changing patterns \cite{benson2010network,roy2015inside}.
Hence, it is unclear if \emph{fixed} network topologies are the right tool for highly optimized networked environments. 

However, even if network designers take traffic demands into their designs, they have little idea of the specific applications that will use the network. Thus, the common approach is to design static and \emph{demand-oblivious} network topologies optimized toward worst-case scenarios like all-to-all communication. However, such an aoblivious pproach is \emph{uncommon} in other computing frameworks; for example, servers use caches (and similar techniques) to enhance  performance and respond to the actual demand.
In recent years, the maturation of optical switching and networks has introduced exciting opportunities for network design. Namely, optical switches that can \emph{dynamically} reconfigure their internal connectivity (input-output port matching), which in turn changes the global circuit-switched  network topology without rewiring or other physical changes to the network. 
The enabling of such \emph{dynamic} switches and dynamic topologies resulted in a flourish of proposals for 
reconfigurable optical networks \cite{hall2021survey}.
These proposals can be divided into two main dimensions.
The first is to keep the reconfiguration \emph{demand-oblivious} and use dynamicity to rotate between predefined topology configurations, e.g., RotorNet \cite{mellette2016scalable,rotornet}, Opera \cite{opera}, and Sirius \cite{ballani2020sirius}. The main advantage of this approach is that rotation can be done fast on a nanosecond scale \cite{ballani2020sirius}.
The second approach, which we focus on more in this paper, is to make the reconfiguration \emph{demand-aware} and adjust the switches based on actual demand in real-time, e.g., Helios \cite{helios}, c-Through \cite{cthrough}, Eclipse, \cite{venkatakrishnan2018costly}, ProjecTor \cite{projector} and others \cite{hall2021survey}. 
The con of this approach is the slower reconfiguration times which are micro-second scale \cite{hall2021survey}.

Interestingly, a recently proposed model, the ToR-Matching-ToR (TMT) model, \cite{cerberus} is a unified model that uses a two layers leaf-spine network architecture and can describe \emph{both} static, dynamic, demand-oblivious, and demand-aware systems. Figure~\ref{fig:TMT} (left) illustrates the TMT model with seven leaf switches (ToR) and three spine switches (each with a matching) (see Section \ref{sec:model} for formal details).

Our work investigates the possibility of a dynamic network topology that can adjust itself to the workload's characteristics. Intuitively, such a network can exploit localities \cite{Denning2005} in the communication patterns and introduce a dynamic topology that optimizes current trends rather than worst-case trends. Thus, dynamic networks can yield shorter routing paths and higher throughput than static networks in highly structured workloads. Specifically, static networks often use expander graphs \cite{xpander} or hierarchical (e.g., FatTree \cite{clos}) topologies to optimize the network \emph{diameter} which give an upper bound for the (average) path length.
While there are several metrics of interest when studying a network's topology, this work focuses on the average path length since a shorter route length leads to better utilization of links and higher network throughput \cite{namyar2021throughput}.  As we will see formally later, our network model provides the designer (i.e., our algorithm) a set of $k$ matchings of size $n$, namely a set of $n k$ direct links that we can reconfigure dynamically. 
A central perspective we examined in this work is to treat our edges similarly to a \emph{links cache} \cite{cacheNet}. 
but, caching network links is 
different than caching arbitrary objects. First, there is a dependency between links in the cache since several consecutive cached links create a \emph{path}, and a collection of links create a cached \emph{graph}. Second, links caches don't have binary hit/miss behavior but incur costs according to the path length between the source and destination. Thus, adding (or removing) a link to the cache may impacts many requests from many sources to many destinations.
\begin{figure}[t]
  \begin{centering}
  \includegraphics[width=\columnwidth]{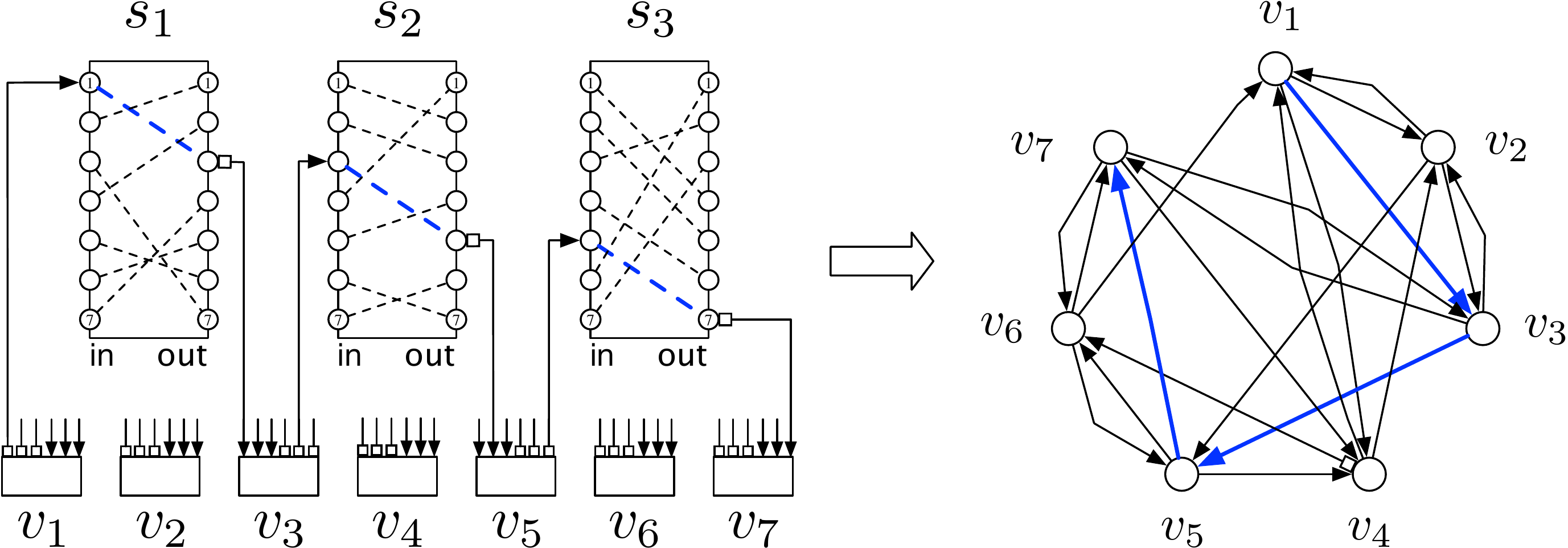} 
    \caption{Example of the leaf-spine \model model (left) with $n=7$ ToR nodes and $k=3$ spine switches (matchings) and the corresponding $k$ regular, directed, network graph at time $t$, $\netw(t)=\bigcup_{i=1}^{k} M(i,t)$ (right).
    The (directed) multi-hop path $v_1, s_1, v_3, s_2, v_5, s_3, v_7$ is shown in the figure. Matching are reconfigured over time within each spine switch.
    } 
    \label{fig:TMT}
  \end{centering}
\end{figure}

\subsubsection*{\bf Our Contribution}
Our work demonstrates that we can reduce the \emph{average path length} (APL) significantly on real network traces using our links cache approach.
While even algorithms that establish only \emph{direct} links, i.e., paths of length one, improve static topologies, we show that regarding the links cache as a \emph{graph} provide stronger benefits. 
In turn, as our main contribution, we proposed a novel online algorithm, \algo, which forms also \emph{indirect} links, looks on the cached links as a graph, and attains a considerably better performance.  We prove several theoretical properties of \algo and evaluate it on eight application traces with feasible parameters within the current technological limitations. Our evaluation demonstrates a consistent reduction in APL, of up to $\approx 60\%$ compared to static expander networks.

The rest of the paper is organized as follows:
In Section \ref{sec:model} we formally present the network model and the metric of interest. In Section \ref{sec:static} we first discuss static solutions and first present \algo.
Next, in Section \ref{sec:online} we discuss \emph{online} and \emph{dynamic} algorithms, and the online version of our algorithms.
Section ~\ref{sec:evaluation} introduces datasets and some further algorithms that we use for evaluation in Section~\ref{sec:results}, where we present our empirical results.
After reviewing related work in~\ref{sec:related} we conclude the paper in
Section~\ref{sec:conclusion} with a short discussion. 



\section{Model and Preliminaries}\label{sec:model}

Our network model is motivated by recent proposals for two layers leaf-spine network architectures in which spine switches support reconfigurable \emph{matching} between input-output ports \cite{rotornet,opera}.
Such architectures are called the ToR-Matching-ToR (\model) model \cite{cerberus} and can model existing systems, e.g., Eclipse \cite{venkatakrishnan2018costly} or ProjecToR \cite{projector}, which relies on a \emph{demand-aware} switches, RotorNet \cite{rotornet}, and Opera \cite{opera}, which rely on a \emph{demand-oblivious} switches with matchings that rotate over time or an optical variant of Xpander~\cite{xpander} which can be built from a collection of static matchings.
Formally, the network interconnects a set of $n$ nodes $\{v_1,v_2, \dots, v_n\}$ (e.g., leafs, ToR switches) using a set of $k$ optical spine switches $S = \{s_1, s_2, \dots, s_k\}$.  Each spine switch has a set of $n$ input and $n$ output ports, internally connected via a directed (i.e., uni-direction) \emph{matching} from the input to the output port. These matchings can be dynamic and change over time.  We denote the matching on switch $i$ at time $t$ by $M(i,t)$. 

Each node (i.e., ToR switch) has $k$ up links and $k$ down links. Given any (leaf) node $v_i$, it's $j$th up port is connected to the $i$th input port of spine switch $s_j$ and its $j$th down port is connected to the $i$th output port of spine switch $s_j$. These links are static and do not change, while, as mentioned, the internal matching inside each switch can change.
See Fig. \ref{fig:TMT} (left) for an example of the \model with seven leaf nodes and three spine switches at a given time $t$.


At each time $t$ our (abstract) network is the union the $k$-matchings, $\netw(t)=\bigcup_{i=1}^{k} M(i,t)$. 
Notice that when all matchings are disjoint \emph{perfect matchings}, having $n$ directed edges, then the resulting network $\netw(t)$ is always a $k$-regular \emph{directed} graph with $nk$ edges. Fig. \ref{fig:TMT} (right) shows an example for 
$\netw(t)$ which is a $3$-regular, directed graph. 
The network $\netw(t)$ supports multi-hop routing during time $t$ where a path of length $2\ell$ on the \model network is always of the form $(v_{i_1}, s_{j_1}, v_{i_2}, s_{j_2}, \cdots , s_{j_{\ell}}, v_{i_{\ell+1}})$ and is translated to a path of length $\ell$ on $\netw(t)$ of the form $(v_{i_1}, v_{i_2}, \cdots, v_{i_{\ell+1}})$.
Fig. \ref{fig:TMT} highlights the path $(v_1, s_1, v_3, s_2, v_5 ,s_3, v_7)$ of length 3 from the source $v_1$ to the destination $v_7$,  both on the \model model (left) and on the network $\netw(t)$ (right).


The network $\netw(t)$ is assumed to serve a workload or network traffic represented as a trace of packets or flowlets \cite{perry2017flowtune}. Formally, a trace $\sigma$ is an ordered sequence of communications requests (e.g., IP packets) $\sigma=((s_1,d_1),(s_2,d_2),(s_3,d_3),\ldots)$,
where $s_t,d_t$ represent the source and destination nodes, respectively, and the request, $(s_{t-1},d_{t-1})$ occurs before the request $(s_t,d_t)$. 

When the $t$th request, $(s_t, d_t)$, from a source $s_t$ to a destination $d_t$ arrives, the cost to serve it is assumed to be proportional to the shortest distance (i.e., number of hops in forwarding the packet) between $s_t$ and $d_t$  on the network $\netw(t)$ which we denote as $\dist_{N(t)}(s_t,d_t)$. 
Recall that our model enables switches to fully reconfigure their connections as long as they form a set of $k$ matchings,
and by that to change $\netw(t)$ over time. 
We assume in our model that switches are restricted to \emph{update} their configuration (matching) only at a predefined rate $\frac{1}{\update}\le 1$, namely $\update$ is a minimum number of consecutive communication requests that are required between two updates. This update rate accounts for the delay needed when changing a configuration in modern switches (see Section \ref{sec:results}).  Therefore, our network is static between configurations (i.e., $\netw(t)=\netw(t+1)$), for a period of $\update$ requests, but the network following a reconfiguration may be completely different (i.e., usually $\netw(t) \neq \netw(t+\update)$). 



A \emph{self-adjusting network} algorithm $\mathcal{A}$ is an \emph{online} algorithm \cite{albers2006online} that selects the $k$-matchings that compose the network at time $t$,  namely, $\netw(t)$. Such \emph{dynamic} algorithms adjust the topology based on some \emph{history} of past requests \cite{avin2019toward} and uses it as an approximation for the near future.  We denote by $\mathcal{A}_{\update}$ an algorithm that is forced to make changes at most once per $\update$ consecutive requests.     
A \emph{static} algorithm is an algorithm that sets the network (i.e., $k$-matching) once and does not change it along the trace. We denote this network as $\netw_0$. An \emph{offline static} algorithm is assumed to know the whole trace $\sigma$ (i.e., the future) when it decides or computes $\netw_0$.  

In turn, our work utilizes the \emph{average path length} ($\apl$) as the cost metric we tries to optimize.
$\apl$ is defined as the cost to serve an entire trace $\sigma$ of length $m=\card{\sigma}$ with respect to 
an algorithm $\mathcal{A}$ and an update rate $\update$. Formally, 
\begin{align}\label{eq:apl}
 \apl(\mathcal{A}_{\update}, \sigma)=\frac{1}{m}\sum_{t=1}^{m}\dist_{N(t)}(s_t,d_t). 
\end{align}

In the next section, we discuss two static algorithms, including our novel proposed algorithm \algo and in Section \ref{sec:online}
we discuss online algorithms.

\section{Static Demand-Aware $k$-Matchings}\label{sec:static}
In this section we explore \emph{static}, \emph{offline} and \emph{demand-aware} algorithms that yield a static network $\netw$ for all the traffic. Next, in Section \ref{sec:online} we study the \emph{online} and \emph{dynamic} version of the problem that constructs a dynamic network $\netw(t)$. 

In the static, offline \emph{demand-aware} network design (DAN) problem \cite{avin2019demandDIST}, we receive a demand distribution $\D$, which describes the frequency (or probability) $\p(u,v)$ of requests between every (directed) pair of nodes in the network. Alternately we can assume that the algorithm receives the trace $\sigma$ as an input, and $\D$ describe the empirical distribution of $\sigma$. Note that since $\D$ is a distribution, we have $\sum \p(u,v) = 1$. The goal of the offline DAN problem is to design a static network (aka a host graph) $\netw \in \nets_k$  which minimizes the \emph{weighted-average path length} where, in our model, we require that $\nets_k$ is the set of all possible networks that are a union of $k$ (directed) matchings. Formally, the $k$-regular DAN problem is:
\begin{align}\label{eq:dan}
\dan(\D) =\argmin_{\netw \in \nets_k} \sum_{(u,v) \in \D} \p(u, v) \cdot \dist_{\netw}(u,v) 
\end{align}

\begin{algorithm}[t]
\caption{\textsc{GreedyMatching}$(\D, k)$ Algorithm}\label{alg:Greedy}
	\begin{algorithmic}[1]
		\Require Demand Matrix $\D$, $k$ - number of switches
		\Ensure $k$-Demand-Aware Matching
    \State Initiate $\netw$ as an empty graph \Comment{{\color{blue} Will be $k$ regular}}
	\State Sort requests in $\D$ by frequency (breaking ties rand.)
	\tikzmk{A}
	\For {each request $(s,d)$} \Comment{{\color{blue} By order of frequency}}
		\If {$\outdeg(s)<k$ and $\indeg(d)<k$}
		\State add $(s,d)$ to $G$
		\EndIf
	\tikzmk{B}
    \boxit{mode2}{1}{1.8}{0.82}{-0.2}
    \EndFor
    \If {$G$ is not $k$ regular and strongly connected}
        \State add (random) edges to make $\netw$, $k$ regular, connected
    \EndIf
    \State Convert $G$ to $k$ matchings
	\end{algorithmic}
\end{algorithm}

Before presenting our algorithm, we first discuss the \emph{greedy $k$ matching} algorithm, a simple, naive, but appealing algorithm to our problem.

\subsection{Selfish Approach: Greedy $k$-Matching}

The weighted $k$-matching problem is an extension of the well-known \emph{weighed matching} problem (i.e., $k=1$) \cite{edmonds1965maximum, duan2010approximating}. As is commonly known, a simple \emph{greedy matching} solves the weighted matching problem with an approximation ratio of $\frac{1}{2}$ \cite{avis1983survey}. It is important to note that the optimization goal of the matching problem (the weight of the matching) is different than that of the $\dan$ problem (minimum average path length). Nevertheless, the problems are related since a maximum matching finds a feasible set of requests (that can be served in a single hop) with the maximum probability mass of in $\D$.  

Therefore, the greedy $k$-matching algorithm follows the same spirit by building a maximum weight $k$ regular directed graph greedily using edges with the largest probabilities in $\D$. It starts by sorting the requests in $\D$ according to their frequencies. Then it greedily adds requests as long as both the source and destination degrees are less than $k$. Algorithm~\ref{alg:Greedy}, \textsc{GreedyMatching}, provides pseudo-code for this approach.

Since our optimization problem is different than the weighted matching problem and allows adding edges that are not in $\D$,
if we do not yet have a $k$-regular directed graph at the end of this phase, our algorithm differentiates from the classical greedy matching that stops and quits.
In contrast, our algorithm continues and adds \emph{random edges} until we get a $k$-regular graph. We use the fact that any $k$-regular directed graph is decomposable to $k$ perfect matchings. Formally, 


\begin{figure*}[t!]
    \centering
    \begin{tabular}{c|c|c}
    \centered{\includegraphics[width=.20\textwidth]{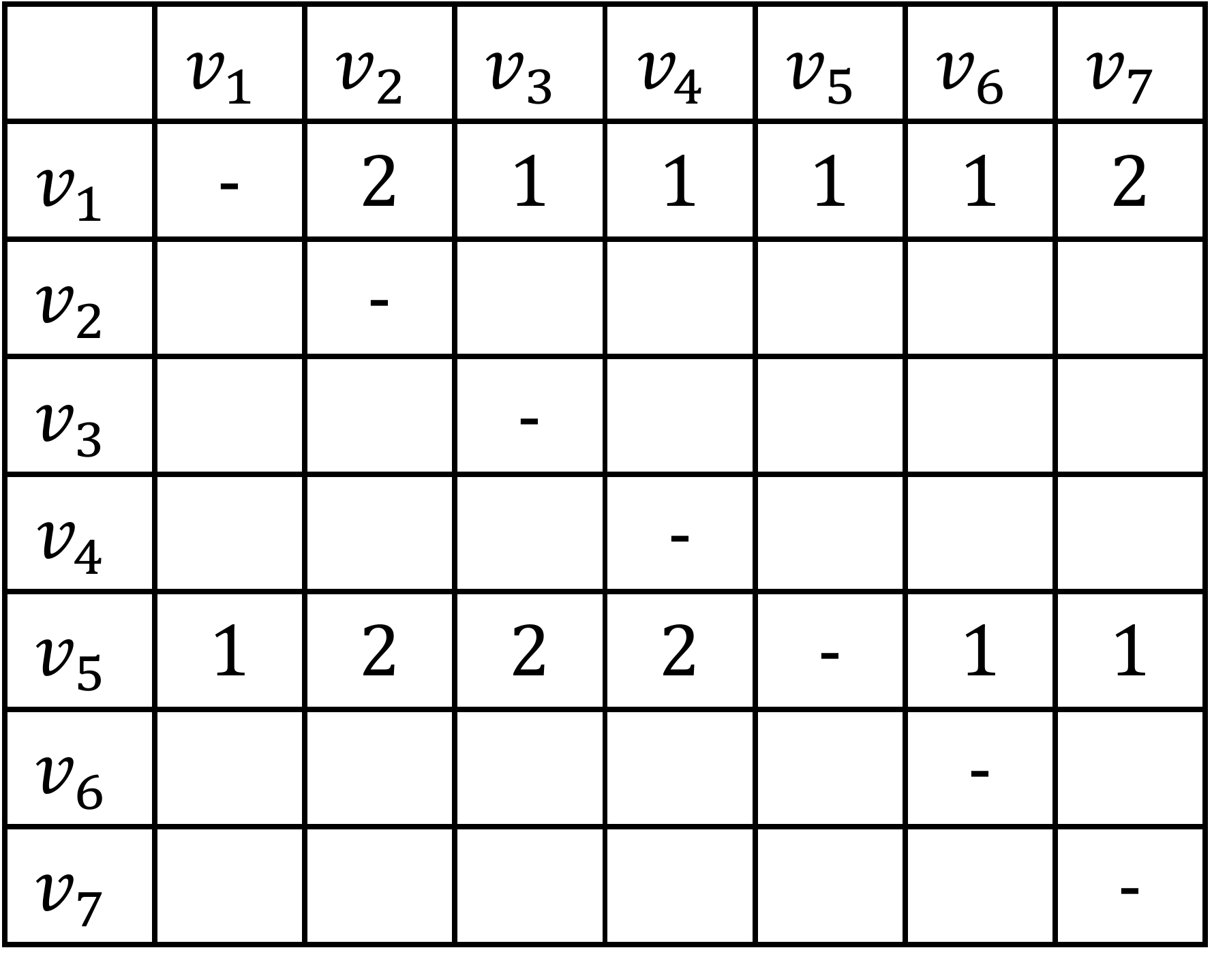}}
    &
    \begin{tabular}{cc}
     \includegraphics[width=.17\textwidth]{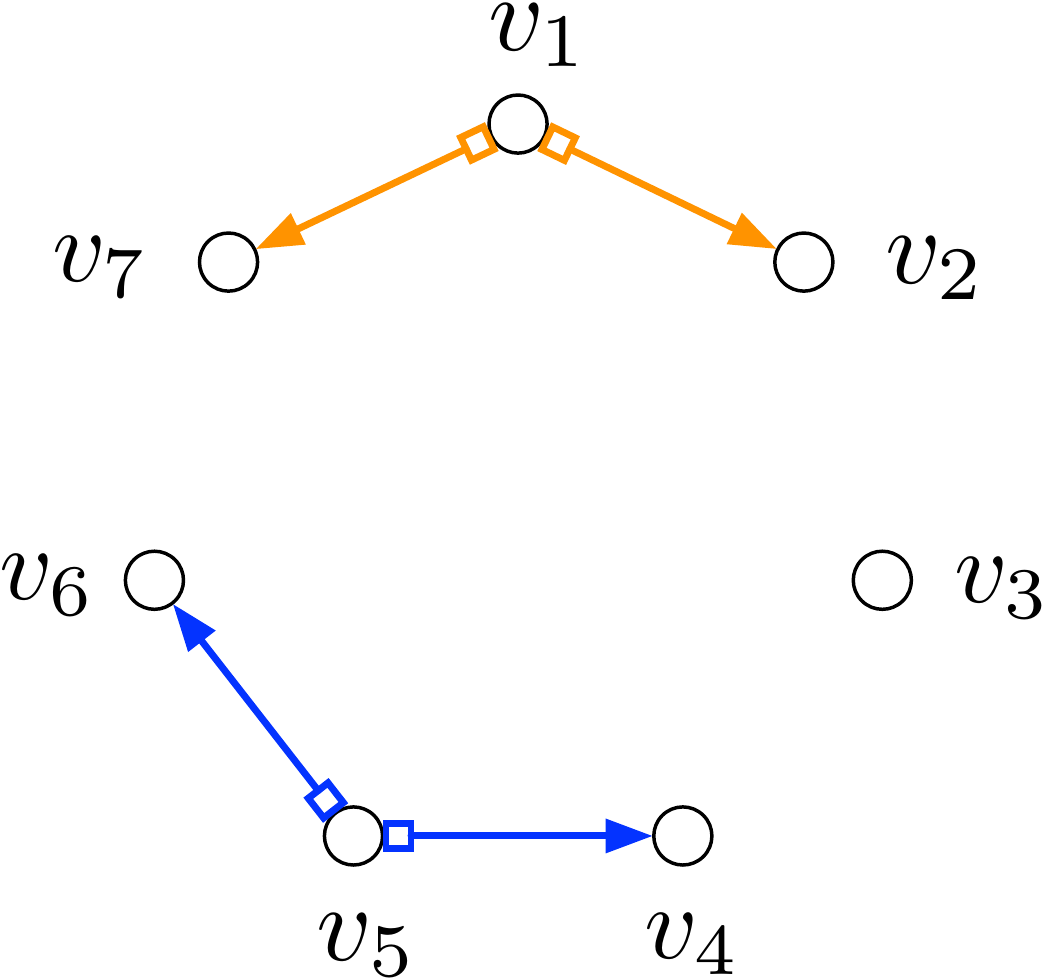} &
     \includegraphics[width=.15\textwidth]{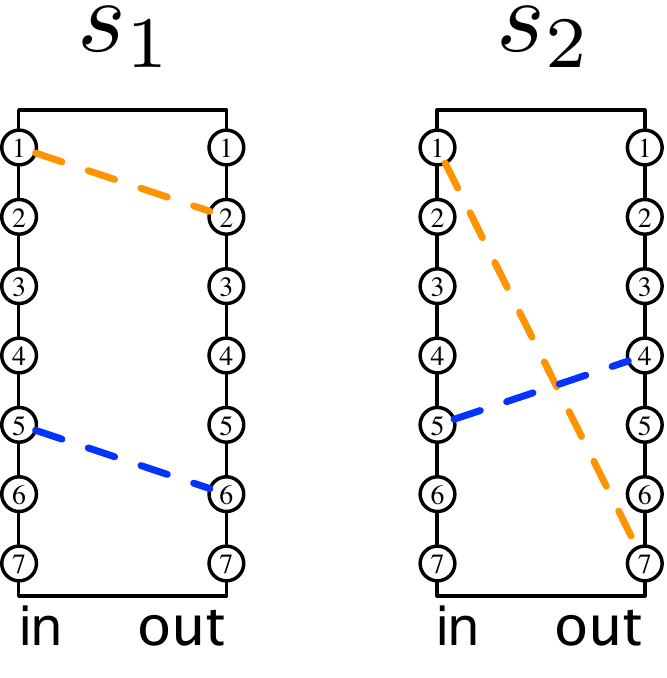}
    \end{tabular}
    &
    \begin{tabular}{cc}
     \includegraphics[width=.17\textwidth]{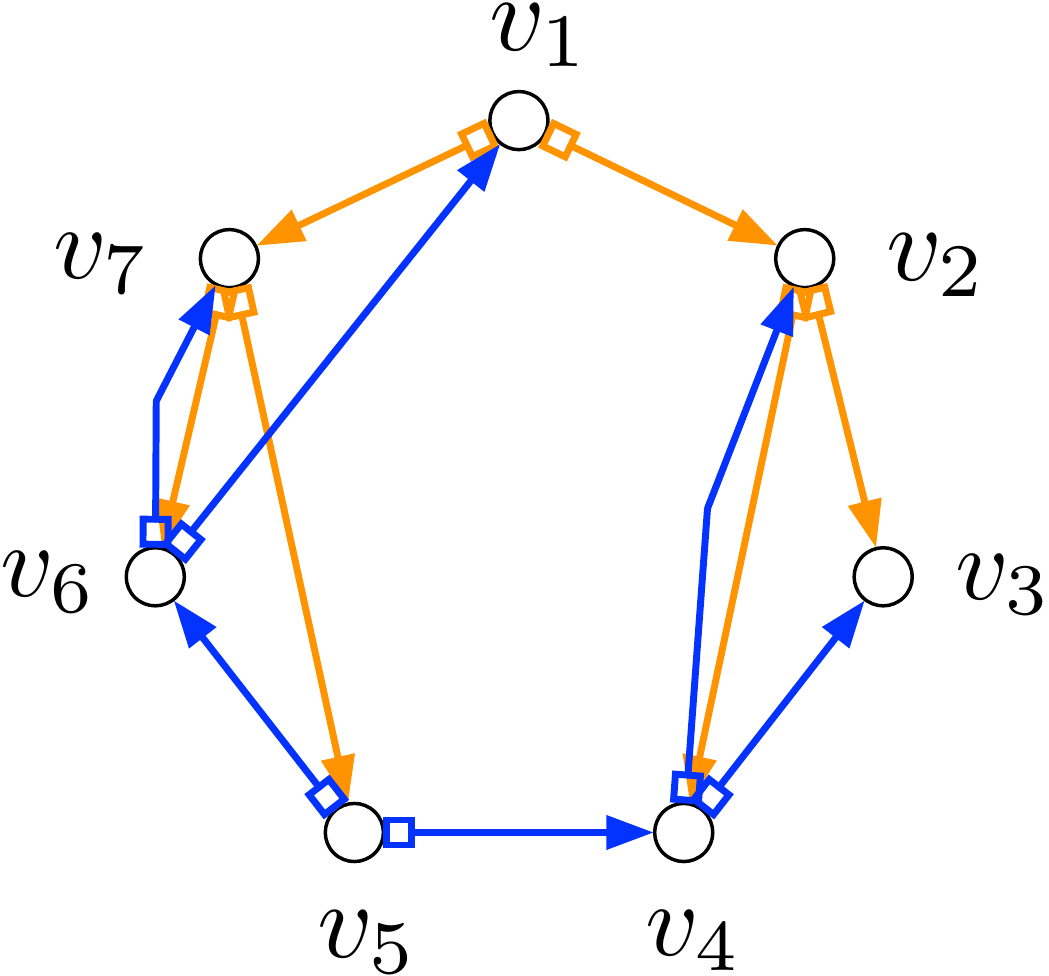} &
     \includegraphics[width=.15\textwidth]{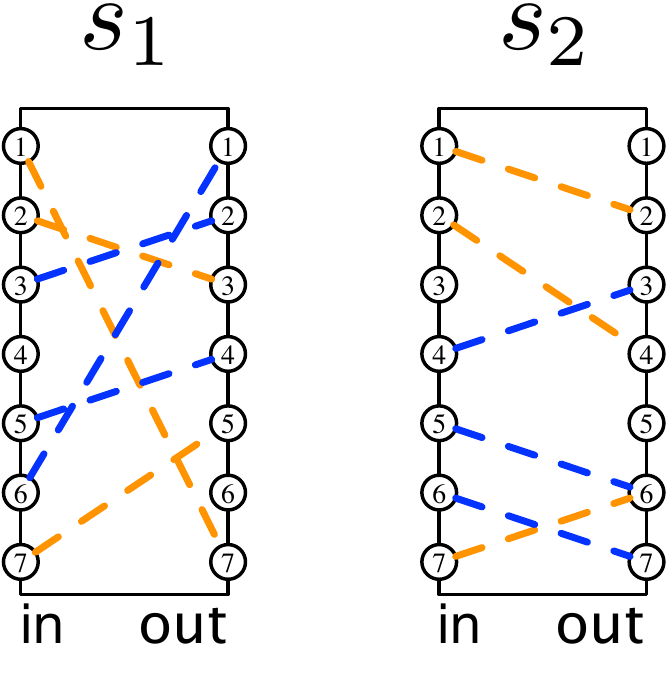}
    \end{tabular}\\
    (a) Demand matrix & (b) The $k$-matching solution & (c) \algo solution  
    \end{tabular}
    \caption{An example of the difference between $k$-matching and \algo. (a) a (directed) demand matrix with a star-like structure, only $v_1$ and $v_5$ communicates to all other nodes. (b) The result of the greedy $2$-matching, shows the network and the corresponding two matching in $s_1$ and $s_2$. (c) The \algo solution for the same demand.}
    \label{fig:get-example}
\end{figure*}

\begin{theorem}\label{thm:decompos}
Any $k$-regular (multi) directed graph can be decomposed to $k$ perfect matchings (one for each of the $k$ switches)
\end{theorem}

The theorem proof follows almost directly because if $G$ is a $k$-regular directed graph, $G$ can be represented as a $k$-regular bipartite graph (by splitting each node to two nodes). In turn, Hall's theorem \cite{hall1935} implies that a perfect matching $M$ exists in $G$. If we remove $M$ from $G$, we are left with a graph $G'$ which is $(k-1)$-regular directed graph, and we can repeat the process, $(k-1)$ more times.

The crucial limitation of the above approach is that it solves the problem from a single, direct link perspective. This approach leads to a selfish behavior where each node only adds edges for its requests.
The selfish approach cannot handle well cases where there is much traffic from a single source to more than $k$ destinations or traffic between more than $k$ sources to a single destination. Such patterns are unfortunately common in real applications. For example, search engines typically partition the search between many destinations. Each searches its local documents and then merges the results, or more generally a map-reduce framework \cite{dean2008mapreduce}. Thus, even if there are many frequently used edges with the same source and varying destinations, \greedy can only select $k$ edges with the same source. 

Fig. \ref{fig:get-example} demonstrates this problem. Fig. \ref{fig:get-example}-(a) present a weighed demand matrix in a \emph{stars} like structure where both $v_1$ and $v_5$ communicates with all other nodes (with different weighed). Fig. \ref{fig:get-example}~(a) shows the solution imposed by the $k$-matching (for $k=2)$. The solution must be a subgraph of $\D$ so only two edges from each star can be included. 

The \algo algorithm we present next overcomes this issue by
taking an altruistic approach and adding indirect paths between sources and destinations using \emph{helper} nodes, in particular other destinations of the same source. Thus, it would still add the most frequent edges in the examples above, but these would not always be direct edges due to topology limitations (of degree at most $k$). Fig. \ref{fig:get-example}-(c) presents the solution of \algo which we discuss in more detail next.


\subsection{Altruistic Approach: Greedy Ego-Trees Network}


We now introduce  \algo, a novel algorithm (see Algorithm \ref{alg:algo} for pseudo-code) to solve the offline DAN problem.
While the basic idea of \algo follows the spirit of similar algorithms like minimum spanning tree (MST) \cite{prim1957} and greedy matching,
\algo brings a new networking (or topology) perspective to the proposed solution and has a theoretical foundation that we discuss later.

To build the network $\netw \in \nets_k$, we first sort the requests in $\D$ according to their frequency. Next, we greedily create paths in the network in the order of the sorted requests until $\netw$ is a $k$ regular directed graph which is then converted to $k$  matchings (Theorem \ref{thm:decompos}) defining the concrete switches configurations.

The key idea of \algo is that when we build a path for a request $(s,d)$, all previous requests (which had higher probabilities) already have a \emph{short} path between their source and destination. Since the in-degree and out-degree of nodes can be at most $k$, we may need intermediate nodes to help us to add a short path between $s$ and $d$ to $\netw$.
We do so by initiating a \emph{forward} Breath-First-Search (BFS) starting at $s$ to find the closet \emph{available} node to $s$, denoted as $x$. That is, we seek for the closest node whose \emph{out-degree} is less than $k$.  Note that initially the closest available node can be $s$ itself. 

Next, we preform a \emph{backward} BFS starting from $d$ to find the closest available node to $d$, denoted as $y$. That is, a node whose \emph{in-degree} is less than $k$, and we can therefore add a new edge $(x,y)$ to the network.
However, we first verify that adding the edge $(x,y)$ results in a shorter path between $s$ and $d$ than the current network. 
Formally, we add the edge $(x,y)$ to $\netw$ only if $\dist(s,x) + \dist(y,d)  +1 < \dist(s,d)$. See Fig. \ref{fig:BFS} for an example of finding $x$ and $y$ using forward and backward BFSs. If no available nodes $x$ or $y$ exist we skip to the next request in the sorted list.

After adding $kn$ directed edges, or adding all the $(s,d)$ requests with $\p(s,d)>0$, if the resulting network is not $k$-regular or if it is not strongly connected, then we need to add or change some of the edges. Specifically, if the network is strongly connected but is not 
$k$-regular, we add random edges between available nodes until the network is $k$ regular. 
If the network is not strongly connected, we identify connected components and connect them by removing low weight (i.e., probability) edges and adding new edges until we reach a single strongly connected component. For simplicity of presentation, we ignore these (solvable) cases and some other minor corner cases (e.g., $x$ or $y$ do not exist) in the pseudo-code description of Algorithm \ref{alg:algo}.
The final step of the algorithm is to decompose the $k$-regular directed network $\netw$ to $k$-matchings. 

Fig. \ref{fig:get-example}~(c) presents the result of 
\algo for the demand matrix in (a). As we can observe, \algo utilizes more edges than the $k$-matching approach. Moreover, this simple example builds an optimal directed ego tree, both for $v_1$ (in green) and $v_5$ (in blue).

\begin{algorithm}[t]
\caption{\textsc{GreedyEgoTrees}$(\D, k)$ Algorithm}\label{alg:algo}
	\begin{algorithmic}[1]
		\Require Demand Matrix $\D$, $k$ - number of switches
		\Ensure $k$-Demand-Aware Matching
    \State Initiate $\netw$ as an empty graph \Comment{{\color{blue} Will be $k$ regular}}
	\State Sort requests in $\D$ by frequency (breaking ties rand.)
	\tikzmk{A}
	\For {each request $(s,d)$} \Comment{{\color{blue} By order of frequency}}
	    \State Let $x$ be an available node in $\Fbfs(s)$
	    \State Let $y$ be an available node in $\Bbfs(d)$
		\If {$\dist(s,x) + \dist(y,d) +1< \dist(s,d)$}
		    \State add $(x,y)$ to $\netw$
		\EndIf
	\tikzmk{B}
    \boxit{mode1}{1}{1.8}{0.82}{-0.2}
    \EndFor
    \If {$G$ is not $k$ regular and strongly connected}
        \State add (random) edges to make $\netw$, $k$ regular, connected
    \EndIf
    \State Convert $\netw$ to $k$ matchings
	\end{algorithmic}
\end{algorithm}

Next, we discuss several theoretical properties of \algo.
We start with the following observation to provide fundamental insights into the motivation behind \algo and the need for less active nodes or edges to help with high-frequency requests. 
We denote by $G_\D$ the weighted directed graph when we see $\D$ as the adjacency matrix of a directed graph.
Consider any demand distribution for which $G_\D$ is a \emph{star network} with a root $r$ (i.e., $r$ is a single source, or single destination, for all requests in $\D$). In this case, \algo will create $\netw$ as a $k$-ary directed tree network where $r$ is the root and nodes' distance from $r$ is ordered by the frequency they communicate with $r$.
It is easy to see that such a $k$-ary tree optimally minimizes the weighted-average route length and that $\netw \in \nets_k$.
Similarly, we can state the following:

\begin{observation}\label{obs:stars}
\algo is optimal for a demand $\D$ for which $G_\D$ is a collection of disjoint (weighted) stars.
\end{observation}

Next, we extend Observation~\ref{obs:stars} to the more general demand distribution $\D$ where $G_\D$ forms a forest and we bound the APL with the \emph{Entropy} \cite{cover2012elements} of the distribution $\D$.

The information \emph{Entropy} (or Shannon entropy) is a measure of the uncertainty, or disorder, in an information source. Since being introduced by Claude Shannon in his seminal 1948 work~\cite{shannon1948mathematical}, entropy has found many uses, including coding, compression, and machine learning to name a few \cite{cover2012elements}. Recently, the conditional entropy was proved to be a lower bound for the average path length in static DAN ~\cite{avin2019demandDIST}. Formally, for a discrete random variable $X$ with possible values
$\{x_1, \dots , x_n\}$, the (base $k$) entropy $H_k(X)$ of $X$ is defined as
\begin{align}
H(X) = \sum_{i=1}^n p(x_i)\log_k\frac{1}{p(x_i)}
\end{align}
where $p(x_i)$ is the
probability that $X$ takes the value $x_i$. Note that, $0 \cdot
\log_k\frac{1}{0}$ is considered as 0.
We can state the following about \algo:

\begin{theorem}For $k>1$ and a distribution $\D$, if $G_\D$ is a directed (weighted) forest, then the weighted-average route length of \algo is less than $H_k(\D) + 1$ where $H_k$ is the entropy (base $k$) function.
\end{theorem}





\begin{figure}[t]
  \begin{centering}
  \includegraphics[width=.6\columnwidth]{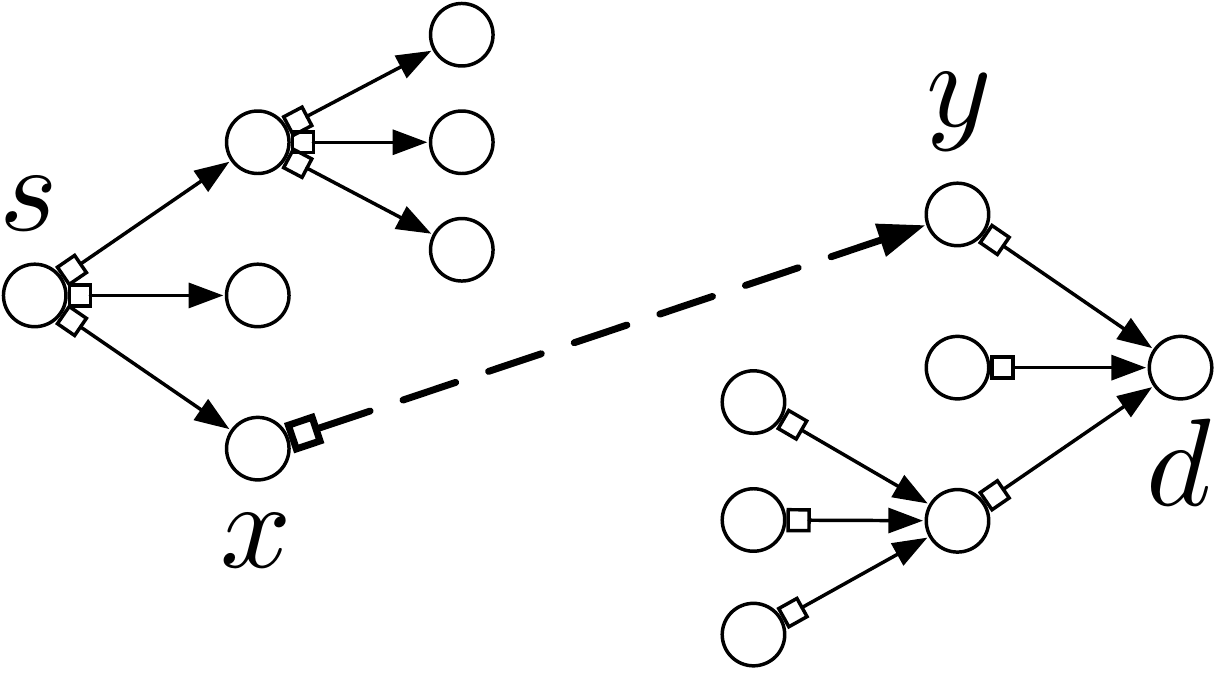} 
    \caption{\algo Algorithm: Forward BFS and backward BFS example with the corresponding $x$ and $y$.}
    \label{fig:BFS}
  \end{centering}
\end{figure}

\begin{proof}[Proof sketch]
First, let $\p_1 \ge p_2 \ge, \dots \ge p_{\card{\D}}$ denote the probabilities of the requests in  $\D$ in a non-increasing order.
Note that it must be the case that $\p_i \le \frac{1}{i}$, otherwise $\sum_{j=1}^i \p_j > 1$, contradiction for $\D$ being a distribution.
Next, we show that for the first $\card{\D}$ edges added to $\netw$, $\netw$ will be a directed forest with max in and out-degree $k$. Consider the $i$th request in the sorted list of requests, $(s_i, d_i)$. Since $G_\D$ is a forest, the $i$th request is the only request in $\D$ for which $d_i$ is a \emph{destination}. Thus, the current in-degree of $d_i$ is zero, and when \algo finishes, its degree will be one. Since the in-degree of all nodes in $\netw$ is at most, and there are no cycles, $\netw$ will also be a forest. By construction, nodes will have an out-degree of at most $k$. 
%
Now consider what will be the distance $\dist(s_i,d_i)$ after adding the edge $(x,y)$
Following Algorithm \ref{alg:algo}, $y=d_i$ and $x$ is the closet node to $s_i$ with out-degree less than $k$. Since $\netw$ is a directed forest, the sub-tree rooted at $s_i$ can have at most $i$ edges and therefore $\dist_{\netw}(s_i,d_i)$ can be at most $\lceil \log_k(i)\rceil$, so $\dist_{\netw}(s_i,d_i) < \log_k i + 1 $. Overall we have,
\begin{align}
    \sum_{i=1}^{\card{\D}} p_i \dist_{\netw}(s_i, d_i) & \le \sum_{i=1}^{\card{\D}} p_i \log_k i +1 \notag \\
    &\le \sum_{i=1}^{\card{\D}} p_i \log_k \frac{1}{\p_i} + 1 \le H_k(\D) + 1
\end{align}
\end{proof}

We note that for the case of general distribution $\D$, the conditional entropy, $H(X|Y)$ is a lower bound for the average path length~\cite{avin2019demandDIST}, where $X, Y$ are the sources and destinations nodes, respectively. Such a lower bound can be potentially much lower than the joint entropy $H(\D) = H(X, Y)$ that we prove above.
Note that after connecting in $\netw$ all pairs from $\D$, the algorithm will add random edges to create a $k$-regular directed graph.

We conclude this section by showing that the running time of \algo is polynomial.   

\begin{theorem}
The running time of \algo is $O(k^2 n^2$).
\end{theorem}

\begin{proof}[Proof overview]
The primary operations in \algo are of polynomial time.
We go over them in the order of the algorithm.
Sorting can be done in $O(nk \log (nk))$.
Next, we need to add one edge, one at a time to $\netw$. Such addition may require the source and destination nodes to construct their BFS tree (forward or backward). BFS runs in $O(m)$ where $m$ is the number of edges. Since we have at most $m=kn$ edges to add, all $m$ BFS searches, two for each edge can be made at a total time of $O(k^2n^2)$.
Additionally, we need to find $\dist_N(s,d)$ in the current network for each path we build. This operation can be done when the source $s$ (or destination $d$) preform the BFS mentioned above. If there is a path from $s$ to $d$, it will be found.
Finally, adding random edges and the graph decomposition (where each maximum matching takes at most $O(nm)=O(kn^2)$) takes no more than $O(k^2n^2)$ as well. 
\end{proof}
We believe that improving the running time of \algo is possible, but leave this question for future work.
In the next section use our static algorithms as building blocks for our discussion on  \emph{online} algorithms.

\section{Online $k$-Regular DAN}\label{sec:online}

 The online DAN problem, denoted as self-adjusting network \cite{avin2019toward}, deals with cases where we do not know the demand matrix ahead of time. Instead, we look at a past \emph{window} of $\win$ requests to approximate the current demand matrix. We use a fixed-sized window to adjust to changes in the demand gradually. 
This section is organized as follows:
Section~\ref{subsec:meta} presents a \emph{meta}-framework for online algorithms for the online DAN problem, Section~\ref{subsec:ego} explains how use matching based algorithms
on top of the framework. 
Notice that all these approaches use the same meta-framework for online algorithms. 

\subsection{Meta-Algorithm for the Online DAN problem}
\label{subsec:meta}
All the online demand-aware network (DAN) algorithms we consider in this work follows the same meta-framework to maintain a dynamic network $N(t)$ and to minimize the $\apl$ according to
Eq. \eqref{eq:apl}. We consider only $k$-regular DAN so $N(t)\in \nets_k$ must be a union of $k$ directed matchings at each time $t$.
Pseudocode for the meta-algorithm is shown in Algorithm~\ref{alg:meta}, and we explain it next.
The algorithm receives a trace $\sigma$, an update rate $\update$ defining the number of requests between subsequent network state updates, and a window size $\win$ used to approximate the current demand matrix $\D$.  In particular at time $t$, $\sigma[t-\win,t]$ denotes the $W$ last request in $\sigma$, an only those can be used to make decisions about the reconfigurations. 
The update rate $\update$ reflects the reconfiguration times imposed by technological limits of optical switches.
Changing matchings takes time and cannot be executed, for example, after each packet.  Therefore, once per $\update$ requests, the algorithm updates the network configuration using the \emph{Update}$()$ function (Line 5). The update function yields a new network configuration $N(t+1)$, i.e., DAN, according to the last $\win$ requests and the current network configuration $N(t)$. 
All the algorithms we study in this work for the online DAN problem follow this meta-algorithm and vary in their implementation of the update function, $\update$, and $\win$. 


\begin{algorithm}[t]
\caption{Meta online matching-update algorithm }\label{alg:meta}
\begin{algorithmic}[1]
		\Require A trace $\sigma$, Update rate $\update$, Window size $\win$
    	\Ensure Dynamic network $N(t) \in \netw_k$
\State  $N(1)=$ Initial network
\For{$t=1$ to  $\card{\sigma}$}
\State Serve $(s_t, d_t)$ on $N(t)$
    \If{$t \equiv 0\pmod \update$} \Comment{An Update, at rate $\update$}
    	\State $N(t+1) = \mathrm{Update}(N(t), \sigma[t-\win,t])$
    \Else
        \State $N(t+1)=N(t)$
    \EndIf
\EndFor
\end{algorithmic}
\end{algorithm}

    

\begin{figure*}[t!]
  \begin{centering}
  \subcaptionbox{MultiGrid}{\includegraphics[width=0.16\textwidth]{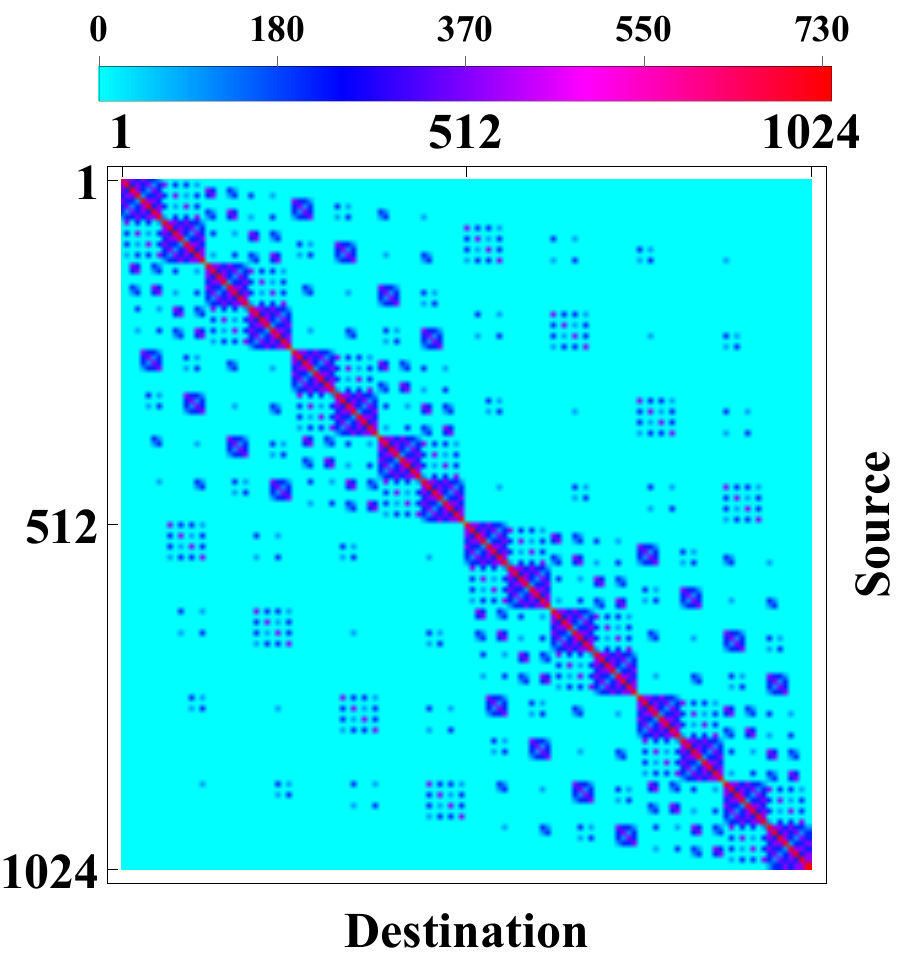}}
  \subcaptionbox{Nekbone}{\includegraphics[width=0.16\textwidth]{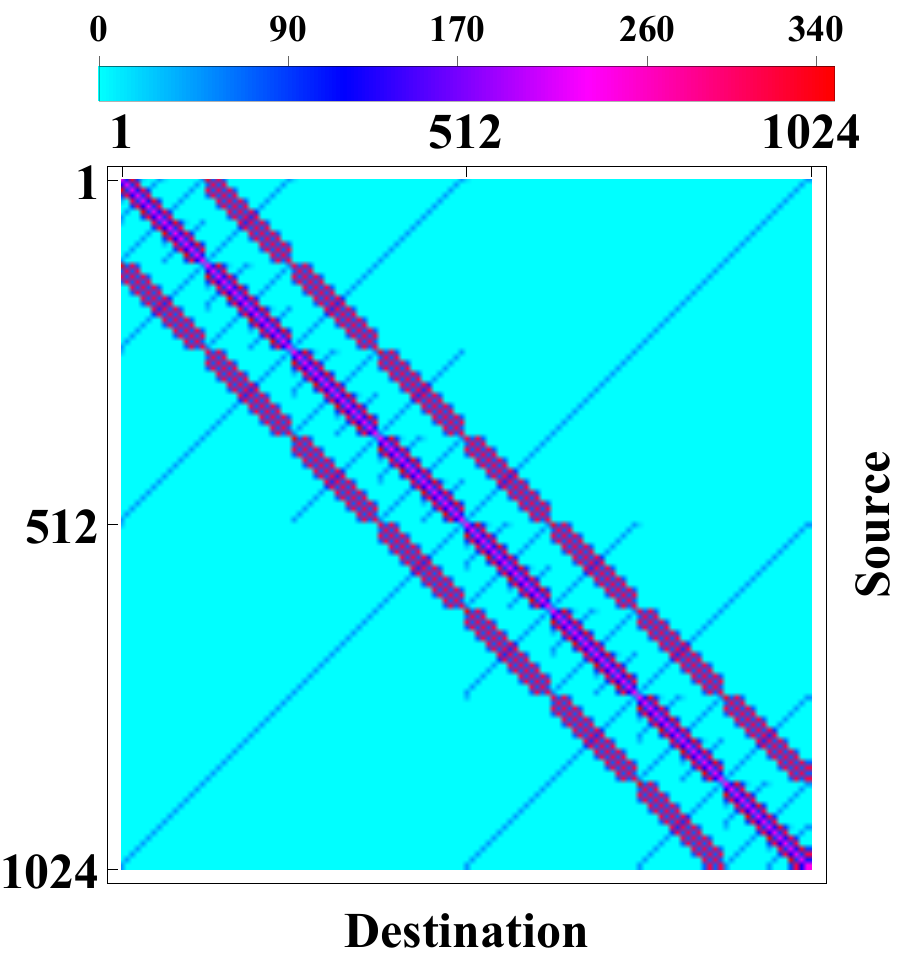}}
  \subcaptionbox{CNS}{\includegraphics[width=0.16\textwidth]{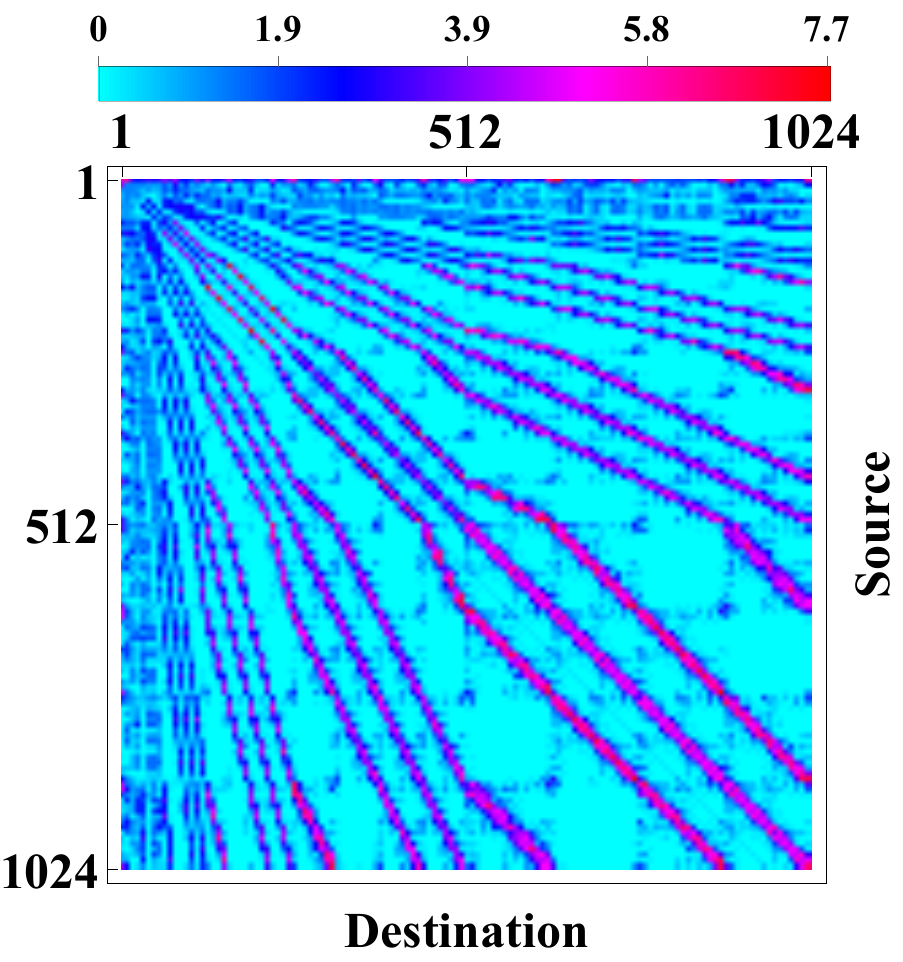}}
    \subcaptionbox{WEB}{\includegraphics[width=0.16\textwidth]{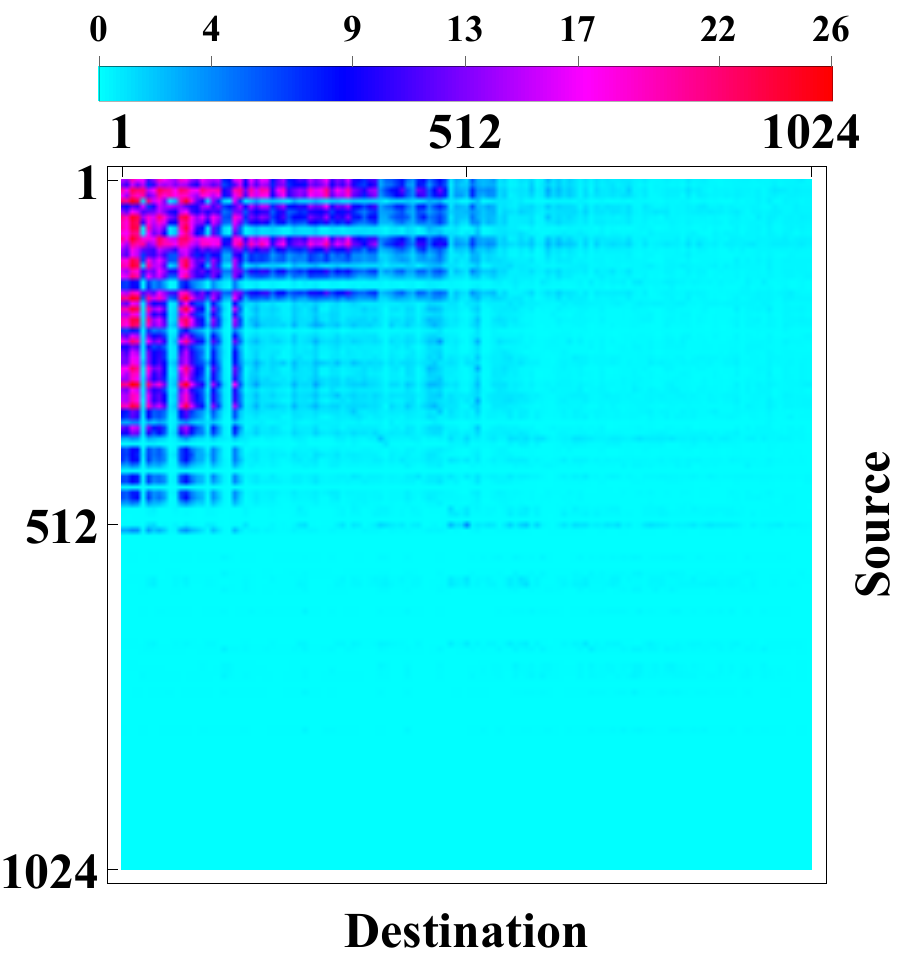}}
  \subcaptionbox{HAD}{\includegraphics[width=0.16\textwidth]{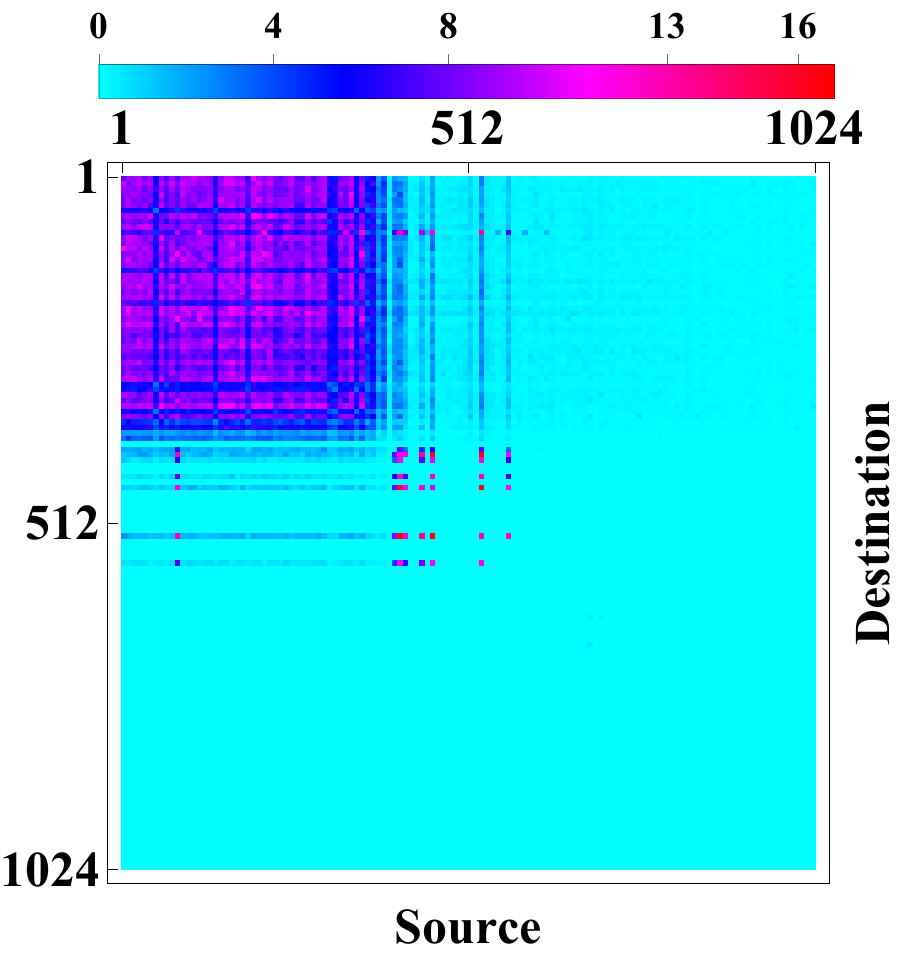}}
    \subcaptionbox{Stars}{\includegraphics[width=0.16\textwidth]{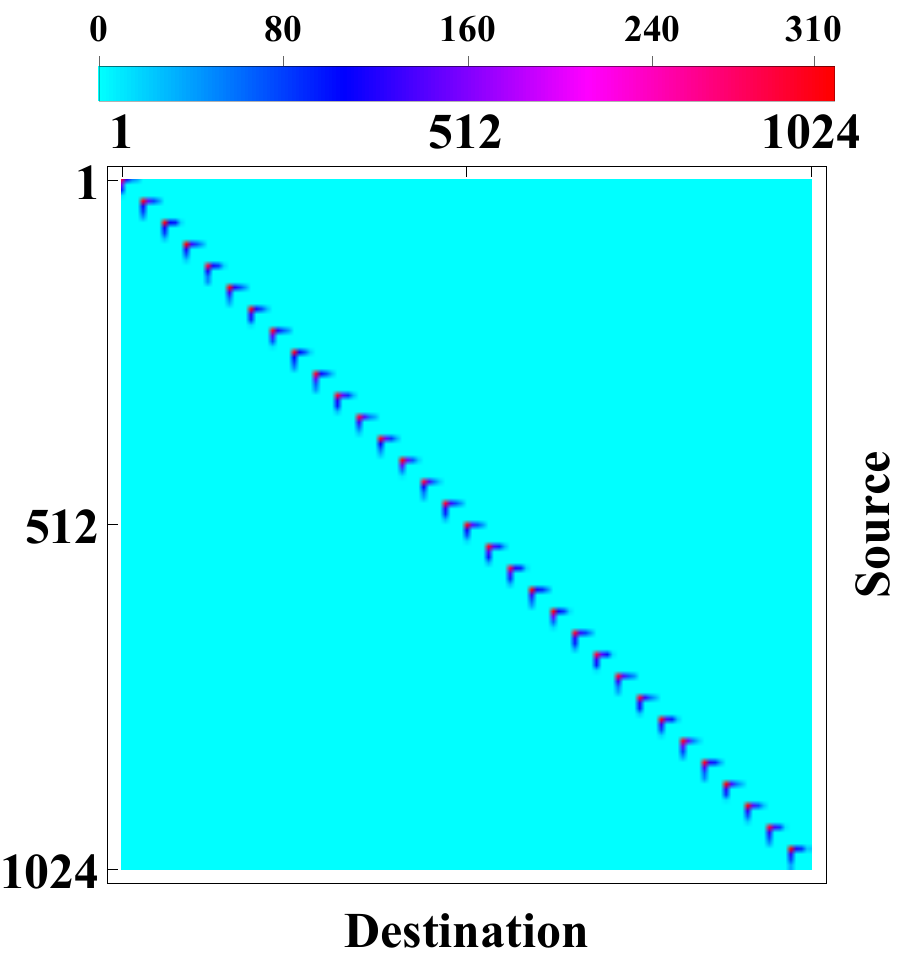}}  
    \caption{Traffic matrices for several of the communication traces. Colors are scaled individually, and the scale is provided at the top of each matrix. Axes represent source IDs (vertical) and destination IDs (horizontal). For HPC and star traces, IDs are given in the order of appearance, and for the Facebook traces, IDs are given by source activity level.}
    \label{fig:matrix}
  \end{centering}
\end{figure*}

\subsection{Online \algo and \textsc{GreedyMatching}}
\label{subsec:ego}

The Online \algo algorithm follows the meta-algorithm. Each time the Update method is called, we prepare a new traffic matrix $D$ based on the requests in $\sigma[t-\win,t]$ and run \algo. Formally 
\begin{align}
N(t+1) = \textsc{GreedyEgoTrees}(\D, k)
\end{align}

In the greedy $k$-matching case we update we use 
\begin{align}
N(t+1) = \textsc{GreedyMatching}(\D, k)
\end{align}

In the evaluation section, we compare these two methods on real traces.
Additionally, we compare both \textsc{GreedyMatching} and \textsc{GreedyEgoTrees}
with a recent proposal for an online b-matching algorithm, \cite{onlineBMatch2021} we describe next.

\subsection{Online b-matching algorithm: (\bma)}
\bma in an online dynamic version of the classic $b$-matching problem \cite{anstee1987polynomial}, originally designed for \emph{undirected} graphs. For a directed graph the problem is identical to the $k$ matchings problem where $b$ is the number of switches.

In \cite{onlineBMatch2021} the authors proposed an online competitive algorithm with an approximation ratio of $O(b)$ (in practical settings). 
We have adopted a directed version of the \bma \cite{onlineBMatch2021} to study in this paper.
\bma uses a links cache of $nk$ edges, and whenever a request arrives at the network, if an edge exists in the links cache, it serves it immediately over a single hop. Otherwise, \bma routes the requests using an alternative static expander network. 
This means that, while the meta-algorithm above uses $nk$ links, \bma uses a topology with a total of at most $2kn$ edges. Another difference from the meta-algorithm is the update rate and the window size. While our algorithms can only update the topology once for every $\update$ requests, \bma uses a cost parameter $\alpha$ to control the update rate. 
A cost of $0$ means that the cache is updated on every request and higher costs decrease the rate of change. In Section \ref{sec:results} we evaluate \bma with $\alpha=6$, same as in \cite{onlineBMatch2021}. This value means that in practice, \bma could change edges at a much faster rate than any of our main algorithms, possibly after only $2\alpha=12$ request. 
\bma works greedily by considering a threshold that depends on $\alpha$. When a source-destination requests reach the threshold, \bma adds that source-destination to the links cache and, if necessary, evicts other edges to keep the degree bounded. For exact details of the algorithm, we refer the reader to the paper.

\section{Datasets and Algorithms}\label{sec:evaluation}

This section introduces the datasets used in this paper and the algorithms we use in our evaluation. 
\subsection{Traffic Traces}
We use eight different traces from three different sources \cite{trace-collection}.
Four traces are from a high-performance computing cluster (HPC), three are from a Facebook (FB) datacenter, and one is a synthetic trace used to present an ideal test case for our algorithm \algo.  Table~\ref{tab:traces_data} provides some high-level relevant statistics such as the length of the traces and certain properties of the demand graph.  These include the number of unique nodes, the number of directed edges and average, and the minimum and max of in- and out- degrees.  We will later use some of these properties to explain the empirical results.  Note that a node that only acts as either a sender or a receiver will have a minimal degree of zero. 

\begin{table}[t]
\caption{Traces}
    \centering
    \begin{tabular}{|l|r|r|r|r|r|r|}
\hline
Name & Length &  Nodes & Edges & Avg~~ &  Min & Max \\ \hline
HPC/MultiGrid   & 1M & 1024 & 21240 & 20.74 & 7 &26      \\ \hline
HPC/Nekbone    & 2M & 1024 & 15461 & 15.1 & 0 & 36           \\ \hline
HPC/Mocfe      & 2.7M & 1024 & 4224 & 4.12 & 0 & 20      \\ \hline
HPC/CNS        & 1M & 1024 & 74308 & 72.56  & 53 &1023          \\ \hline
FB/DB   & 1M & 1024 & 84159 & 82.18 & 0 &825      \\ \hline
FB/WEB    & 1M & 1024 & 99301 & 96.97 & 0 & 639           \\ \hline
FB/HAD      & 0.8M & 1024 & 154275 & 150.65 & 0 & 577      \\ \hline
Synth/Stars        & 1M & 1024 & 1984 & 2  & 1 & 31          \\ \hline
\end{tabular} \vspace{-0.4cm}%
\label{tab:traces_data}
\end{table}

\subsubsection*{HPC traces}\emph{Four} traces of exascale applications in high-performance computing (HPC) clusters~\cite{doe2016characterization}. We refer to these traces as  MultiGrid, Nekbone, MOCFE, and CNS, each of these, represents a different application. Fig. \ref{fig:matrix} (a) through (c) provides additional intuition about the traffic patterns of these traces, as they show clear patterns. Some are more ordered than others.


\subsubsection*{Facebook Traces}
This set contains three different datacenter clusters from Facebook~\cite{roy2015inside}. The original traces represent more than 300M requests each (including self-loops), with different entries at different aggregation levels, such as pods, racks, and host IPs. The traces we used in this paper are sub traces that contain $1024$ nodes of the most active pairs of the racks level source-destination address trace.   
The three different clusters represent three different application types, Hadoop (HAD), a Hadoop cluster, web (WEB), front-end web cluster that serves web traffic, and database (DB), MySQL servers.


\subsubsection*{Stars Trace}
The Star trace is the (only) synthetic trace we use in this work, and we chose it to demonstrate the ideal patterns for \algo. Star contains a demand trace from a set of disjoint star graphs of the same size. That is, requests only travel from a star's center to its leaves or vice versa. We used a Zipf-like distribution \cite{reed2001pareto} to determine the traffic distribution in each star. That is, $\frac{C}{i}$ is the probability of a packet to or from the $i$'th leaf, and $C$ is a normalization constant.  
Specifically, our trace uses  $32$ stars each with $31$ leaves resulting in exactly $1024$ nodes like the real traces. Figure~\ref{fig:matrix} (f) shows the traffic matrix of this trace.

\subsection{Tested Algorithms}
Our evaluation includes the three \emph{online} algorithms presented in Section \ref{sec:online}, Namely online greedy~$k$-matching, \bma and our proposed algorithm online \algo. 

We also compare the online algorithm with dynamic topologies to two static topologies as a baseline. Specifically, we consider two \emph{static} networks built from $k$ static matching:  i) a \emph{demand-oblivious} expander graph  ii) a \emph{demand-aware} offline \algo algorithm which assumes the knowledge of the whole trace. 

\subsubsection*{Demand-oblivious expander network}
Since we are interested in short average path length networks, a natural solution considers networks with short diameters. Expanders~\cite{hoory2006expander} that were recently suggested as a datacenter topology \cite{xpander} are well-known regular graphs with good properties, including large expansion, multiple disjoint paths, small mixing times, and short diameter of  $O(\log n)$ where $n$ is the number of vertices. 

We construct our expander network by creating $k$ uniformly at random matching, one for each switch. It is known from previous work that expanders can be created by taking the union of a few matchings~\cite{goldreich2011basic}. We created many expanders and selected the best one when considering $\apl$ on an all-to-all communication pattern. 

\subsubsection*{Demand-aware offline \algo}
Here we consider the \algo algorithm when the whole trace $\sigma$ is known as an input. Recall that \algo creates a demand-aware $k$ regular network, and in this case, sees the \emph{future} requests. Thus it is an interesting benchmark for online \algo, indicating how important it is to change configuration dynamically and capture temporal communication patterns within the trace. 


\section{Empirical Results}\label{sec:results}

In this section, we evaluate the algorithms on the traces dataset. 
We start by exploring the effects of the window size ($\win$) and update rate ($\update$) on the average path length (APL). 
In our evaluation we have tested an array of window sizes and update rates. The tested values for both 
$\win$ and $\update$ were $\{0.5 \PLH 10^4 ,1 \PLH 10^4 ,2 \PLH 10^4 ,4 \PLH 10^4 ,10 \PLH 10^4 \}$ packets.
This provided us with a total of $25$ tests per trace. 

\begin{table}[t]
\renewcommand{\arraystretch}{1.2}
\small
\centering
\begin{tabular}{|l|r|r|c|}
\hline
Trace Name & Best $\update$ &  Best $\win$ & APL Diff  \\ \hline
\multicolumn{4}{c}{\textbf{Online \algo}}         \\ \hline
HPC/MultiGrid    & $0.5  \PLH 10^4$ & $2  \PLH 10^4$ & 0.4 \%          \\ \hline
Synth/Stars     &  $10  \PLH 10^4$ & $10  \PLH 10^4$ & 2.1 \%    \\ \hline
FB/HAD  &$10  \PLH 10^4$ & $ 4  \PLH 10^4$ & 6.5 \%   \\ \hline
\multicolumn{4}{c}{\textbf{Online \greedy}}        \\ \hline
HPC/MultiGrid     &  $1 \PLH 10^4$ & $2 \PLH 10^4$ & 0  \%          \\ \hline
Synth/Stars      &  $10  \PLH 10^4$ & $10  \PLH 10^4$ & 0.1 \%    \\ \hline
FB/HAD         &  $10  \PLH 10^4$ & $2  \PLH 10^4$ & 2 \%        \\ \hline
\end{tabular} 
\caption{Window size and update rate. Simulation setup are 
update rate $\update=1 \PLH 10^4$, window size $\win =2 \PLH 10^4$.}
\label{tab:ParamesTable}
\end{table}

\subsection{The Window size}
We now discuss the window size ($\win$) parameter used in \algo and \greedy. These algorithms use ($\win$) to estimate the current demand matrix, and it is unclear what is the ideal window size. A short window may be desirable when the demand matrix changes significantly over time. When demand is (relatively)  static, larger windows yield a better result as they more accurately sample the current demand.

Table~\ref{tab:ParamesTable} indicated what were the best $\win$ and $\update$ parameters for three different examples traces. That is, the $\win, \update$ combination that produced the lowest APL for both of our main online algorithms: \algo and \greedy. The last column in the table shows the difference in ($\%$) from the APL produced when using our chosen parameters, which we will discuss shortly.  
When looking at the window size column at Table \ref{tab:ParamesTable} we notice that most optimal values lie at $2 \PLH 10^4$, one exception is the Stars synthetic trace, where the largest tested window size yielded the best result, though by small margin of $2.1\%$ and $0.1\%$ for online \algo and \greedy respectively. This is predictable since a small window does not capture all request types. The lack of temporal locality in the stars trace means that there is no benefit in using a small windows for this synthetic trace. This is unlikely to be true for most other traces as they usually contain some temporal locality \cite{tracecomplexity}.
The HAD trace is intriguing, as \greedy favors a smaller window than \algo. Intuitively, \greedy can only build a smaller number of paths, and thus in small windows, it can exploit temporal-locality. However, \algo creates many more (indirect) paths and benefits from a larger window size. 


In conclusion, our evaluation demonstrates that the desired window size depends on the workload and the algorithm. However, from now on, we use $\win =2\PLH10^4$ since, as we can observe in the difference column of Table~\ref{tab:ParamesTable}, that value is empirically satisfactory in all the workloads yielding at most a small difference in APL performance.

 \begin{figure}
  \begin{centering}
  \includegraphics[width=1\columnwidth]{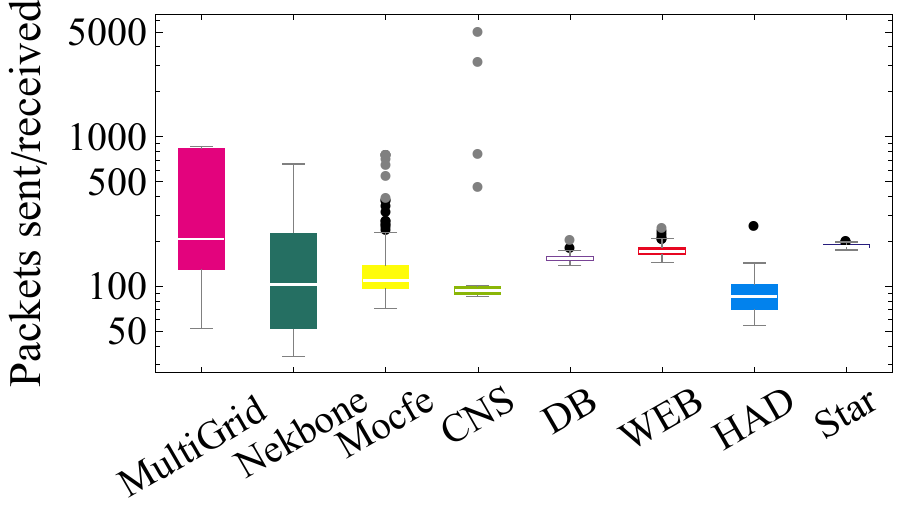} 
  \end{centering}
  \caption{A box plot representing the distribution of most active nodes source/destination inside each update window, using an update rate of $\update = 10^4$. Note that the $y$ axis is in log scale. }
    \label{fig:boxChart}
  \end{figure}

\begin{figure*}[t!]
  \begin{centering}
 \setlength{\tabcolsep}{3pt}
 \begin{tabular}{rcccccccc}
    &   \multicolumn{8}{c}{\includegraphics[height=.7cm]{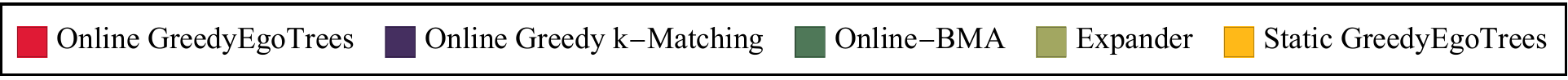}} \\
  \raisebox{1cm}{\rotatebox[origin=lc]{90}{\small Average Path Length}} \hspace{-0.3cm} &
  \subcaptionbox{Mocfe}{\includegraphics[width=0.11\textwidth]{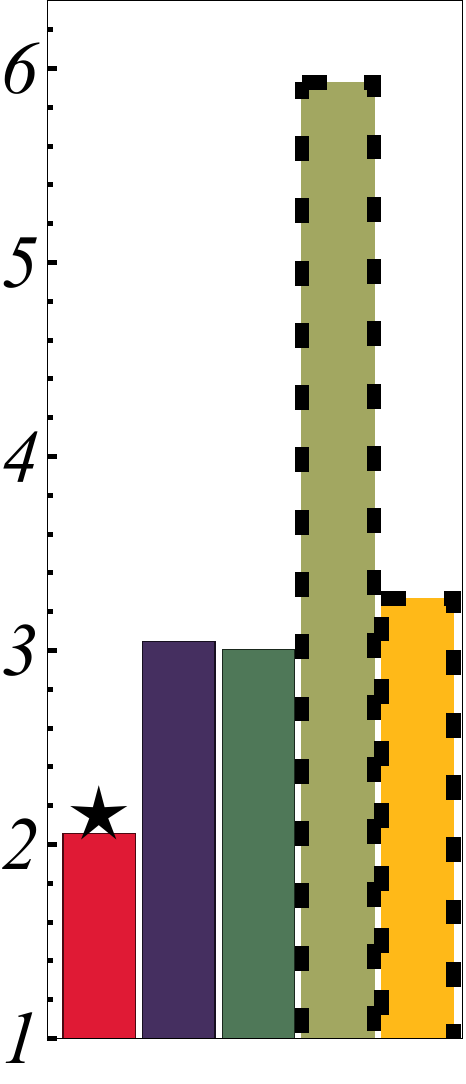}} &
  \subcaptionbox{Stars}{\includegraphics[width=0.11\textwidth]{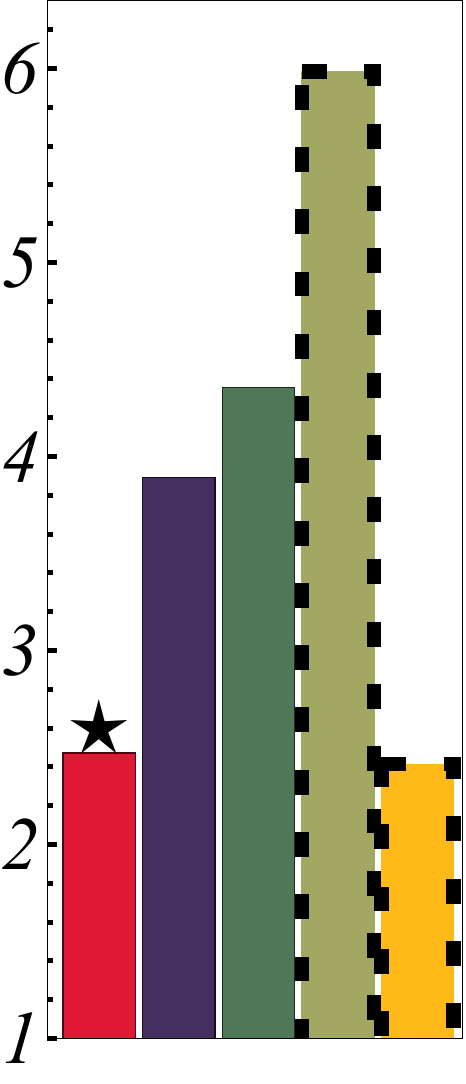}} &
  \subcaptionbox{MultiGrid}{\includegraphics[width=0.11\textwidth]{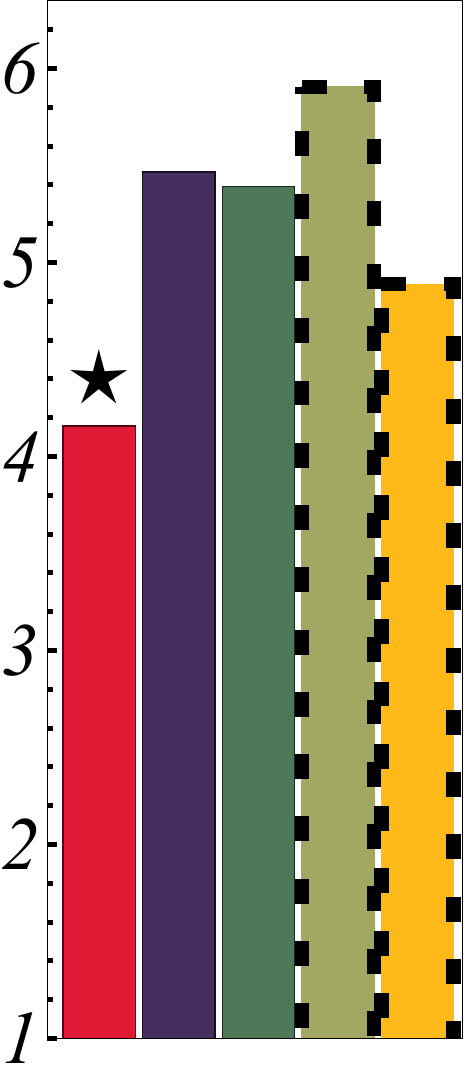}} &
  \subcaptionbox{Nekbone}{\includegraphics[width=0.11\textwidth]{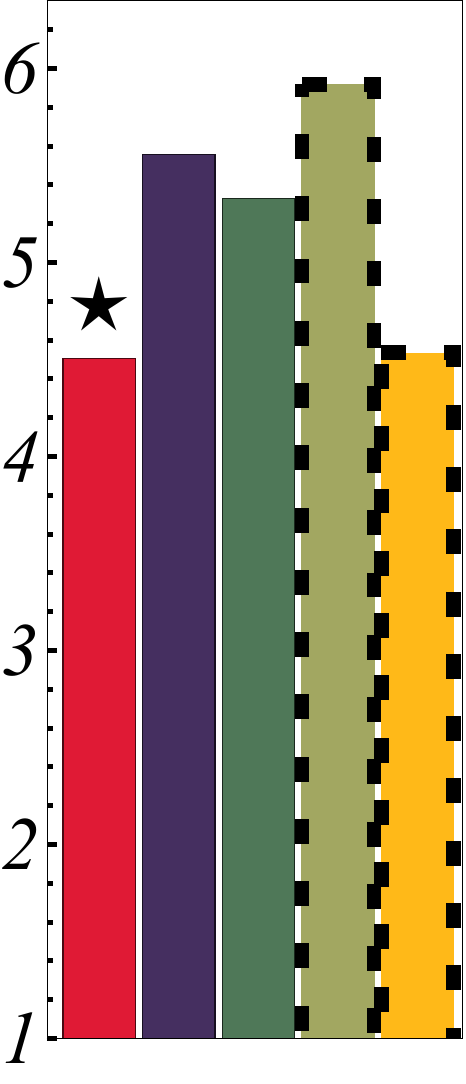}} &
  \subcaptionbox{WEB}{\includegraphics[width=0.11\textwidth]{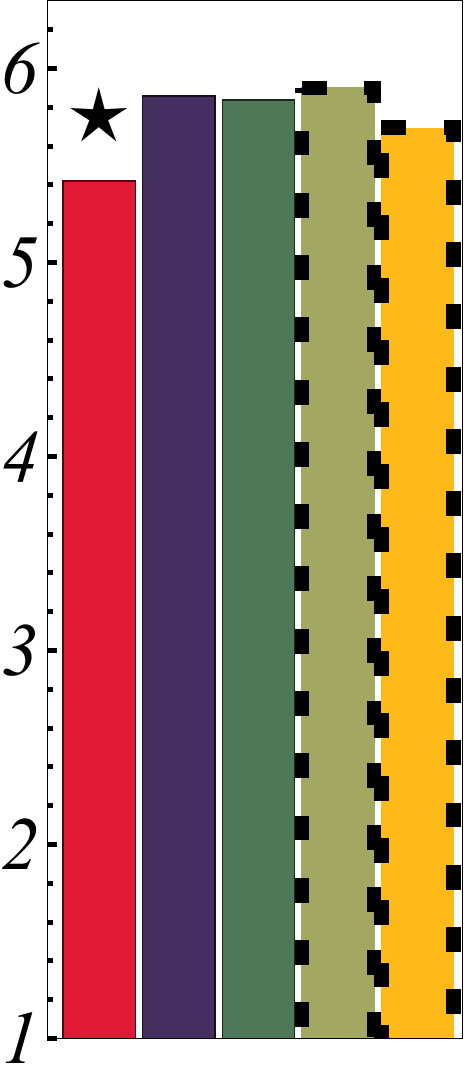}} &
    \subcaptionbox{HAD}{\includegraphics[width=0.11\textwidth]{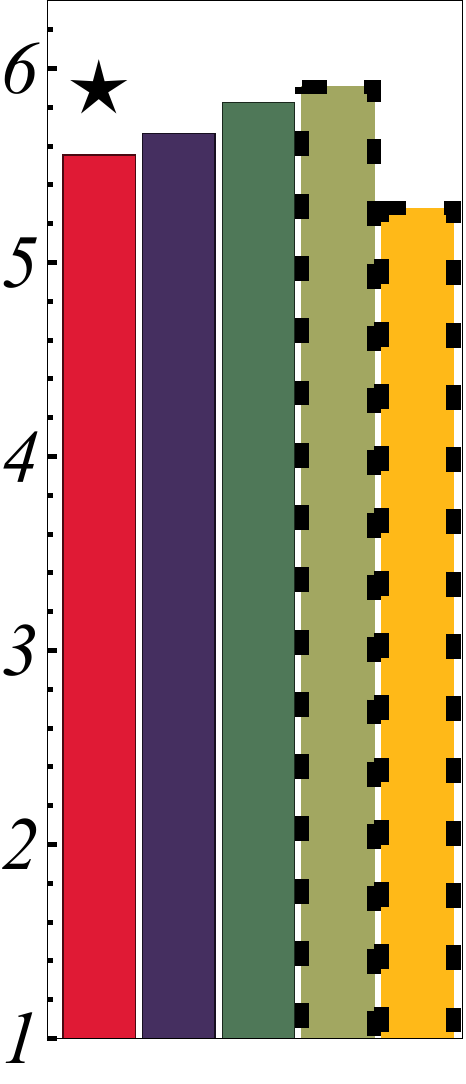}} &
  \subcaptionbox{DB}{\includegraphics[width=0.11\textwidth]{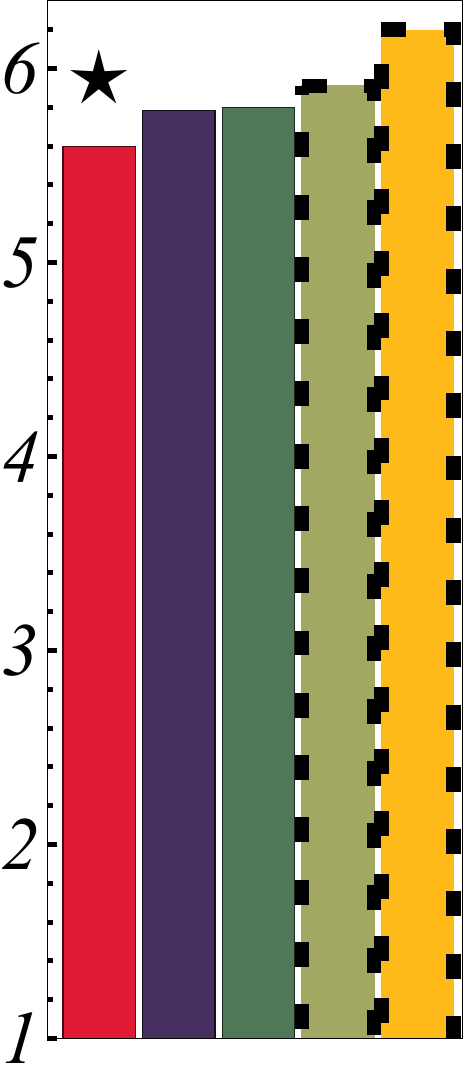}} &
  \subcaptionbox{CNS}{\includegraphics[width=0.11\textwidth]{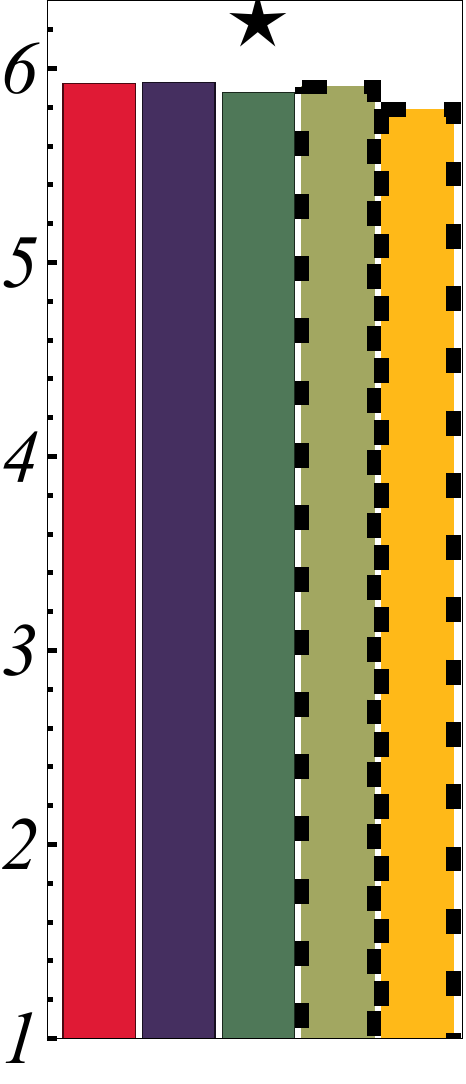}} 
   \end{tabular}
    \caption{$\apl$ for static algorithms (dashed lines) and dynamic algorithms (full lines) on eight traces.  We use $\update=10^4$, and $\win=2  \PLH 10^4$ when applicable. The star sign ({\fontfamily{lmr}\selectfont$\star$}) marks the best \emph{dynamic} algorithm  for each trace. }
    \label{fig:mainresults}
  \end{centering}
\end{figure*}

\subsection{The Update Rate}
The update rate ($\update$) controls the update frequency, and we expect faster updates to yield better results. However, real-world constraints on reconfiguration time would prevent real systems from employing high-speed update rates. Specifically, the reconfiguration time of state-of-the-art optical switches is circa a few tens of microseconds ($\mu s$)~\cite{hall2021survey,mordia}. Due to a lack of accurate timing regarding packets in our traces, our work measures the update rate in packet counts. We may determine that a realistic reconfiguration time would be anything above once per $10^4$ packets according to the following back of the envelope calculation: transmitting a single MTU-sized packet (1500B) on a single $40 Gps$ takes roughly $0.3\mu s$. Therefore, we can send approximately $100$ such packets on a port within the reconfiguration time of network switches. 
For $10^4$ packets to represent a window of time at least $30\mu s$ long, we would need at least one node to send/receive $100$ packets or more during the window. In this case, even if all $100$ packets were sent back to back with no delay from the beginning of the reconfiguration window, we would know that enough time has elapsed. Figure~\ref{fig:boxChart} explores how common is such an event. For each of our $8$ traces, the a box and whiskers represent the distribution of the number of packet sent/received by the most active node in the window of $10^4$ packets, for the whole trace. We can see that for every trace other than for HAD, the average number of packets is at least $100$. Furthermore, the distributions show the majority of windows contain at least one high-volume node. 
Considering a more empirical approach, looking at the update rate column at Table \ref{tab:ParamesTable} we notice two cases, for MultiGird, the optimal update rate is low, $10^4$ or $0.5  \PLH10^4$, while the other two traces, HAD and stars find an optimal value at $10 \PLH 10^4$.  Again, we can attribute this to the existence, or lack thereof of temporal locality in a trace. 


We conclude that while a low update rate is desired, it is not always necessary for a good result. Thus, we continue our evaluation with an update interval of $10^4$ requests, which we believe matches current technology capabilities, and, as we can observe in the difference column of Table~\ref{tab:ParamesTable}, it is empirically satisfactory in all the workloads yielding at most a small difference in APL performance. 

\subsection{The Average Path Length (APL)}

This section discusses the major metric of interest, the average path length of each dynamic or static topology.
The dynamic algorithms use an update rate of $\update =1\PLH 10^4$ and a window of $\win =2\PLH 10^4$ (where relevant). We run the tests on all eight traces described in Table~\ref{tab:traces_data} and five algorithms: online greedy~$k$-matching, \bma, online \algo, static expander, and static \algo.

Figure~\ref{fig:mainresults} presents the APL for all traces and algorithms (topologies). Full lines mark dynamic algorithms in the figure, and dashed bars mark static algorithms. The figure uses black stars to indicate which dynamic algorithm has the best performance (lowest APL). Also note, that the bar charts are ordered from the lowest to the highest APL for online \algo. We note that the result does not consider the first window in each trace.   

We start with the observation that the expander graph has APL of roughly $5.9$ for all traces.
This is not surprising since the expander is demand oblivious, and on expectation will be the same for all traces (not aware of the expander topology). This also shows that demand oblivious networks may reach consistent results, but fail to take advantage of structure in the traffic pattern and are optimized for worst case traces, which lack structure. 
In contrast, the static \algo, which is demand aware is almost always better than the expander and changes between traces. But, recall that Static \algo has advance knowledge of the entire trace.
One exception is the DB trace, where Static \algo is slightly worse than the expander, possibly due to overfitting the network to the heaviest $kn$ edges. This overfitting is mitigated for the online version of the algorithm, showing the benefits of dynamic networks. 

Compared to the dynamic online algorithms, static \algo displays mixed results. It is slightly better than all dynamic algorithms in a few cases, such as for Stars and HAD, and in others (most notably the Mocfe trace), it is much worse. For reference, the Stars trace has no temporal locality, and indeed Static \algo is better on that trace. These results highlight that dynamic demand-aware networks, such as those represented by online \algo manages to benefit from temporal-locality, found in these traces, and can beat static algorithms even those  which are clairvoyant. 

Let us compare now the different online algorithms. In this evaluation, the two matching-based algorithms \bma and  \greedy attain very similar results in all but the synthetic Star trace. A possible explanation is that the Star trace has no temporal locality for the algorithms to exploit. Figure ~\ref{fig:mainresults} also reveals that \algo has the lowest APL among all the dynamic algorithms with one exception on the CNS trace, where the result is slightly in favor of \bma. Recall that \bma uses more edges as it has an additional expander network. In any case, \algo comes second by less than a $1\%$ in this case. Note also that for CNS there is only a negligible improvement with \algo compared to the other algorithms, static or otherwise. 
Therefore, the benefits of \algo could depend on the average degree and the amount of structure found in the trace. These results show that our multi-hop based approach beats the more common k-matching, single hop approach.
We can look into a distribution of path lengths taken by packets on our online \algo and online \greedy to helps us understand how the former beats the latter. Figure \ref{fig:CDFPathLen} shows a CDF of path lengths served by either Online \algo or Online \greedy on two traces Mocfe and Nekbone. While Online \greedy has more requests sent over a single hop for both traces, a consequence of it optimizing towards this result, Online \algo can optimize towards overall shorter path lengths by sacrificing single-hop connections. Furthermore, on the Mocfe trace, online \algo can send more than $70\%$ of requests with a path length of $2$ or less, showing how this algorithm can take advantage of the structure in the trace and optimize towards lower APL.
To conclude, the results show that \algo is better and at least no worse than any of the other dynamic algorithms. It most notably outperforms static expander (random graph), representing a current best demand-oblivious topology. The results demonstrate locality patterns in real application traces, and that \algo manages to leverage these to yield shorter routing paths. Also, observe that \algo is better than \bma, which is, in turn, better than demand-oblivious expanders. Our evaluation demonstrates that there are opportunities in demand-aware networks, particularly in dynamic demand-aware networks. 

\begin{figure}[t]
  \begin{centering}
  \includegraphics[width=1\columnwidth]{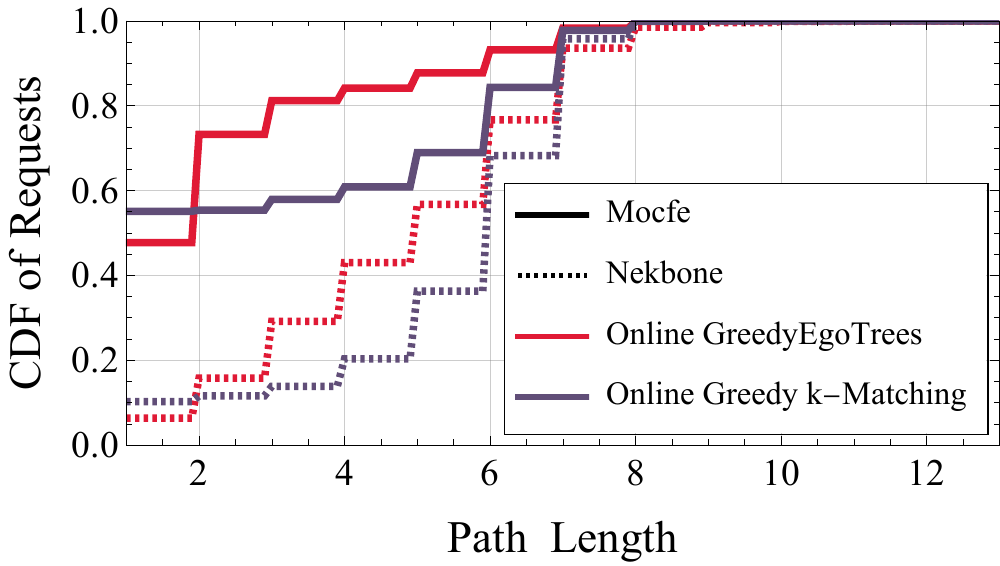} 
  \end{centering}
  \caption{A CDF of path lengths when served by either Online \algo (red) or Online \greedy (purple) on two traces Mocfe and Nekbone, represented by the solid and dashed lines respectively.}
  \label{fig:CDFPathLen}
\end{figure}

\section{Related Work}\label{sec:related}
Today, multi-rooted fat-trees and Clos topologies are some of the most widely deployed static datacenter networks \cite{clos,jupiter,f10}. More recently, Expander graph-based topologies become an important type of static topology which is being evaluated~\cite{xpander,jellyfish}.
Static topologies are naturally demand-oblivious, but not all demand-oblivious networks are static. For example, RotorNet~\cite{rotornet}, Opera~\cite{opera} and Sirius \cite{ballani2020sirius} are dynamic optical datacenter network designs that are demand-oblivious. 
They work by rotating between oblivious and predefined matchings, emulating a complete graph network, and thus providing high throughput. 
However, these designs introduce some difficulties in terms of synchronization and routing. Moreover, they don't change the topology in a demand-aware fashion, and some recent work suggests that demand-aware dynamic topologies may have an advantage over demand-oblivious networks in terms of throughput \cite{cerberus}. 


A connection between demand-aware network designs and information theory is established in
\cite{avin2019toward,splaynet,avin2019demandDIST}. These studies show that achievable routing performance relies on the (conditional) entropy of the demand. Furthermore, empirical studies point towards traffic patterns in datacenters being often \emph{skewed} and \emph{sparse} \cite{projector,roy2015inside}. And an analysis of datacenter traffic in the form of trace complexity shows that network traffic contains patterns that can be leveraged to improve network performance \cite{tracecomplexity}. Another recent work, Cerbrus~\cite{cerberus}, has shown that dynamic and static demand-oblivious networks can be augmented with a demand-aware network. The combined network can outperform each of the pure demand-oblivious networks.
Meanwhile, several dynamic demand-aware datacenter network designs were proposed, including, Eclipse \cite{venkatakrishnan2018costly}, Mordia~\cite{mordia}, or Solstice~\cite{solstice}. 
These suggestions employ traffic matrix scheduling via Birkhoff-von-Neumann decomposition. The generated schedule of single-hop connection serves all demands in an ideal way, with no limit on reconfiguration. Other projects, such Helios~\cite{helios}, ProjecToR ~\cite{projector}, and Online dynamic $b$-matching \cite{onlineBMatch2021} 
focus on maximum matching algorithms. None of these designs use indirect multi-hop routing on their dynamic infrastructure. 
In particular, $b$-matching for undirected graphs is quite similar to our approach of online greedy $k$-matching, as it optimizes for the highest cost matching. They present the \bma algorithm, which is shown to be a constant factor approximation. 
Regrading the offline $k$-matching (edge-disjoint) problem (in undirected graphs), in \cite{hanauer2022fast} the authors attempt to find an optimal heavy matching, using different offline algorithms, which entail higher running times that depend on $k$. They also show the $k$-matching problem is NP-hard for $k \ge 3$. 
We note that the $b$-matching and the $k$ edge-disjoint matchings for undirected graphs are not the same problem. For example, the three edges of a triangle can be covered with a $2$-matching, but not with a $k$ edge-disjoint matchings.  
CacheNet~\cite{cacheNet} have recently offered to model demand-aware networks as a network of cached links and compared this network to the demand oblivious RotorNet~\cite{rotornet}. However, while CacheNet is similar to our approach, it differs in several ways. First, the optimized metric is different. Second, the links cache in CacheNet is of \emph{unbounded} degree, and third, CacheNet uses only single-hop routing.
We are not aware of any work exploring an online demand-aware algorithm similar to \algo and believe that our interpretation of links cache in bounded degree networks is novel.

\section{Conclusions}\label{sec:conclusion}
Our work demonstrates that a demand-aware network design can further optimize the network topology and reduce the average path length. We present \algo{} that successfully leverages temporal and nontemporal localities in workloads, yielding a shorter average path length than static expander-based networks and previous demand-aware algorithms. 
Specifically, our online \algo forms short routing paths according to a dynamic demand matrix. Through extensive evaluations, we show that \algo attains up to 60\% reduction in average path length with respect to the static expander networks. 

Looking into the future, we seek to form more dynamic algorithms that adapt their configuration (most notably the window size and request weights) to the current workload. 




\bibliographystyle{./bibliography/IEEEtran}
\balance
\bibliography{literature}

\begin{thebibliography}{10}
\providecommand{\url}[1]{#1}
\csname url@samestyle\endcsname
\providecommand{\newblock}{\relax}
\providecommand{\bibinfo}[2]{#2}
\providecommand{\BIBentrySTDinterwordspacing}{\spaceskip=0pt\relax}
\providecommand{\BIBentryALTinterwordstretchfactor}{4}
\providecommand{\BIBentryALTinterwordspacing}{\spaceskip=\fontdimen2\font plus
\BIBentryALTinterwordstretchfactor\fontdimen3\font minus
  \fontdimen4\font\relax}
\providecommand{\BIBforeignlanguage}[2]{{%
\expandafter\ifx\csname l@#1\endcsname\relax
\typeout{** WARNING: IEEEtran.bst: No hyphenation pattern has been}%
\typeout{** loaded for the language `#1'. Using the pattern for}%
\typeout{** the default language instead.}%
\else
\language=\csname l@#1\endcsname
\fi
#2}}
\providecommand{\BIBdecl}{\relax}
\BIBdecl

\bibitem{OECD20Keeping}
\BIBentryALTinterwordspacing
OECD, ``Keeping the internet up and running in times of crisis,'' 2020.
  [Online]. Available:
  \url{https://www.oecd-ilibrary.org/content/paper/4017c4c9-en}
\BIBentrySTDinterwordspacing

\bibitem{rotornet}
W.~M. Mellette, R.~McGuinness, A.~Roy, A.~Forencich, G.~Papen, A.~C. Snoeren,
  and G.~Porter, ``Rotornet: A scalable, low-complexity, optical datacenter
  network,'' in \emph{Proc. of the ACM SIGCOMM Conference}, 2017, pp. 267--280.

\bibitem{ballani2020sirius}
H.~Ballani, P.~Costa, R.~Behrendt, D.~Cletheroe, I.~Haller, K.~Jozwik,
  F.~Karinou, S.~Lange, K.~Shi, B.~Thomsen \emph{et~al.}, ``Sirius: A flat
  datacenter network with nanosecond optical switching,'' in \emph{Proc. of the
  ACM SIGCOMM Conference}, 2020, pp. 782--797.

\bibitem{flexspander}
M.~Y. Teh, Z.~Wu, and K.~Bergman, ``Flexspander: augmenting expander networks
  in high-performance systems with optical bandwidth steering,'' \emph{IEEE/OSA
  Journal of Optical Communications and Networking}, vol.~12, no.~4, pp.
  B44--B54, 2020.

\bibitem{xpander}
S.~Kassing, A.~Valadarsky, G.~Shahaf, M.~Schapira, and A.~Singla, ``Beyond
  fat-trees without antennae, mirrors, and disco-balls,'' in \emph{Proc. of the
  ACM SIGCOMM Conference}.\hskip 1em plus 0.5em minus 0.4em\relax ACM, 2017,
  pp. 281--294.

\bibitem{clos}
M.~Al-Fares, A.~Loukissas, and A.~Vahdat, ``A scalable, commodity data center
  network architecture,'' in \emph{ACM SIGCOMM CCR}, vol.~38, no.~4.\hskip 1em
  plus 0.5em minus 0.4em\relax ACM, 2008, pp. 63--74.

\bibitem{benson2010network}
T.~Benson, A.~Akella, and D.~A. Maltz, ``Network traffic characteristics of
  data centers in the wild,'' in \emph{Proc. of the 10th ACM SIGCOMM conference
  on Internet measurement}.\hskip 1em plus 0.5em minus 0.4em\relax ACM, 2010,
  pp. 267--280.

\bibitem{roy2015inside}
A.~Roy, H.~Zeng, J.~Bagga, G.~Porter, and A.~C. Snoeren, ``Inside the social
  network's (datacenter) network,'' in \emph{Proc. ACM SIGCOMM CCR}, vol.~45,
  no.~4.\hskip 1em plus 0.5em minus 0.4em\relax ACM, 2015, pp. 123--137.

\bibitem{hall2021survey}
M.~N. Hall, K.-T. Foerster, S.~Schmid, and R.~Durairajan, ``A survey of
  reconfigurable optical networks,'' \emph{Optical Switching and Networking},
  p. 100621, 2021.

\bibitem{mellette2016scalable}
W.~M. Mellette, G.~M. Schuster, G.~Porter, G.~Papen, and J.~E. Ford, ``A
  scalable, partially configurable optical switch for data center networks,''
  \emph{Journal of Lightwave Technology}, vol.~35, no.~2, pp. 136--144, 2016.

\bibitem{opera}
W.~M. Mellette, R.~Das, Y.~Guo, R.~McGuinness, A.~C. Snoeren, and G.~Porter,
  ``Expanding across time to deliver bandwidth efficiency and low latency,'' in
  \emph{Proc. of USENIX NSDI}, 2020, pp. 1--18.

\bibitem{helios}
N.~Farrington, G.~Porter, S.~Radhakrishnan, H.~H. Bazzaz, V.~Subramanya,
  Y.~Fainman, G.~Papen, and A.~Vahdat, ``Helios: a hybrid electrical/optical
  switch architecture for modular data centers,'' \emph{ACM SIGCOMM CCR},
  vol.~41, no.~4, pp. 339--350, 2011.

\bibitem{cthrough}
G.~Wang, D.~G. Andersen, M.~Kaminsky, K.~Papagiannaki, T.~Ng, M.~Kozuch, and
  M.~Ryan, ``c-through: Part-time optics in data centers,'' \emph{ACM SIGCOMM
  CCR}, vol.~41, no.~4, pp. 327--338, 2011.

\bibitem{venkatakrishnan2018costly}
S.~B. Venkatakrishnan, M.~Alizadeh, and P.~Viswanath, ``Costly circuits,
  submodular schedules and approximate carath{\'e}odory theorems,''
  \emph{Queueing Systems}, vol.~88, no. 3-4, pp. 311--347, 2018.

\bibitem{projector}
M.~Ghobadi, R.~Mahajan, A.~Phanishayee, N.~Devanur, J.~Kulkarni, G.~Ranade,
  P.-A. Blanche, H.~Rastegarfar, M.~Glick, and D.~Kilper, ``Projector: Agile
  reconfigurable data center interconnect,'' in \emph{Proc. of the ACM SIGCOMM
  Conference}, 2016, pp. 216--229.

\bibitem{cerberus}
C.~Griner, J.~Zerwas, A.~Blenk, M.~Ghobadi, S.~Schmid, and C.~Avin, ``Cerberus:
  The power of choices in datacenter topology design-a throughput
  perspective,'' \emph{Proc. of the ACM on Measurement and Analysis of
  Computing Systems}, vol.~5, no.~3, pp. 1--33, 2021.

\bibitem{Denning2005}
P.~J. Denning and P.~J., ``{The locality principle},'' \emph{Communications of
  the ACM}, vol.~48, no.~7, p.~19, jul 2005.

\bibitem{namyar2021throughput}
P.~Namyar, S.~Supittayapornpong, M.~Zhang, M.~Yu, and R.~Govindan, ``A
  throughput-centric view of the performance of datacenter topologies,'' in
  \emph{Proc. of the ACM SIGCOMM Conference}, 2021, pp. 349--369.

\bibitem{cacheNet}
C.~Griner, C.~Avin, and S.~Schmid, ``Cachenet: Leveraging the principle of
  locality in reconfigurable network design,'' \emph{Computer Networks}, p.
  108648, 2021.

\bibitem{perry2017flowtune}
J.~Perry, H.~Balakrishnan, and D.~Shah, ``Flowtune: Flowlet control for
  datacenter networks,'' in \emph{Proc. of USENIX NSDI}, 2017, pp. 421--435.

\bibitem{albers2006online}
S.~Albers, ``Online algorithms,'' in \emph{Interactive Computation}.\hskip 1em
  plus 0.5em minus 0.4em\relax Springer, 2006, pp. 143--164.

\bibitem{avin2019toward}
C.~Avin and S.~Schmid, ``Toward demand-aware networking: A theory for
  self-adjusting networks,'' \emph{ACM SIGCOMM CCR}, vol.~48, no.~5, pp.
  31--40, 2019.

\bibitem{avin2019demandDIST}
C.~Avin, K.~Mondal, and S.~Schmid, ``Demand-aware network designs of bounded
  degree,'' \emph{Distributed Computing}, pp. 1--15, 2019.

\bibitem{edmonds1965maximum}
J.~Edmonds, ``Maximum matching and a polyhedron with 0, 1-vertices,''
  \emph{Journal of research of the National Bureau of Standards B}, vol.~69,
  no. 125-130, pp. 55--56, 1965.

\bibitem{duan2010approximating}
R.~Duan and S.~Pettie, ``Approximating maximum weight matching in near-linear
  time,'' in \emph{2010 IEEE 51st Annual Symposium on Foundations of Computer
  Science}.\hskip 1em plus 0.5em minus 0.4em\relax IEEE, 2010, pp. 673--682.

\bibitem{avis1983survey}
D.~Avis, ``A survey of heuristics for the weighted matching problem,''
  \emph{Networks}, vol.~13, no.~4, pp. 475--493, 1983.

\bibitem{hall1935}
P.~Hall, ``On representatives of subsets,'' \emph{Journal of the London
  Mathematical Society}, vol. s1-10, no.~1, pp. 26--30, 1935.

\bibitem{dean2008mapreduce}
J.~Dean and S.~Ghemawat, ``Mapreduce: simplified data processing on large
  clusters,'' \emph{Communications of the ACM}, vol.~51, no.~1, pp. 107--113,
  2008.

\bibitem{prim1957}
R.~C. Prim, ``Shortest connection networks and some generalizations,''
  \emph{The Bell System Technical Journal}, vol.~36, no.~6, pp. 1389--1401,
  1957.

\bibitem{cover2012elements}
T.~M. Cover and J.~A. Thomas, \emph{Elements of information theory}.\hskip 1em
  plus 0.5em minus 0.4em\relax John Wiley \& Sons, 2012.

\bibitem{shannon1948mathematical}
C.~E. Shannon, ``A mathematical theory of communication,'' \emph{Bell system
  technical journal}, vol.~27, no.~3, pp. 379--423, 1948.

\bibitem{onlineBMatch2021}
M.~Bienkowski, D.~Fuchssteiner, J.~Marcinkowski, and S.~Schmid, ``Online
  dynamic b-matching: With applications to reconfigurable datacenter
  networks,'' \emph{ACM SIGMETRICS Performance Evaluation Review}, vol.~48,
  no.~3, pp. 99--108, 2021.

\bibitem{anstee1987polynomial}
R.~P. Anstee, ``A polynomial algorithm for b-matchings: an alternative
  approach,'' \emph{Information Processing Letters}, vol.~24, no.~3, pp.
  153--157, 1987.

\bibitem{trace-collection}
\BIBentryALTinterwordspacing
Trace collection. [Online]. Available: \url{https://trace-collection.net/}
\BIBentrySTDinterwordspacing

\bibitem{doe2016characterization}
U.~DOE, ``Characterization of the {DOE} mini-apps,''
  \url{https://portal.nersc.gov/project/CAL/doe-miniapps.htm}, 2016.

\bibitem{reed2001pareto}
W.~J. Reed, ``The pareto, zipf and other power laws,'' \emph{Economics
  letters}, vol.~74, no.~1, pp. 15--19, 2001.

\bibitem{hoory2006expander}
S.~Hoory, N.~Linial, and A.~Wigderson, ``Expander graphs and their
  applications,'' \emph{Bulletin of the American Mathematical Society},
  vol.~43, no.~4, pp. 439--561, 2006.

\bibitem{goldreich2011basic}
O.~Goldreich, ``Basic facts about expander graphs,'' in \emph{Studies in
  Complexity and Cryptography. Miscellanea on the Interplay between Randomness
  and Computation}.\hskip 1em plus 0.5em minus 0.4em\relax Springer, 2011, pp.
  451--464.

\bibitem{tracecomplexity}
C.~Avin, M.~Ghobadi, C.~Griner, and S.~Schmid, ``On the complexity of traffic
  traces and implications,'' \emph{Proc. of the ACM on Measurement and Analysis
  of Computing Systems}, vol.~4, no.~1, pp. 1--29, 2020.

\bibitem{mordia}
G.~Porter, R.~Strong, N.~Farrington, A.~Forencich, P.~Chen-Sun, T.~Rosing,
  Y.~Fainman, G.~Papen, and A.~Vahdat, ``Integrating microsecond circuit
  switching into the data center,'' in \emph{Proc. of the ACM SIGCOMM
  Conference}, 2013, pp. 447--458.

\bibitem{jupiter}
A.~Singh, J.~Ong, A.~Agarwal, G.~Anderson, A.~Armistead, R.~Bannon, S.~Boving,
  G.~Desai, B.~Felderman, P.~Germano \emph{et~al.}, ``Jupiter rising: A decade
  of clos topologies and centralized control in google's datacenter network,''
  \emph{ACM SIGCOMM CCR}, vol.~45, no.~4, pp. 183--197, 2015.

\bibitem{f10}
V.~Liu, D.~Halperin, A.~Krishnamurthy, and T.~Anderson, ``F10: A fault-tolerant
  engineered network,'' in \emph{Proc. of USENIX NSDI}, 2013, pp. 399--412.

\bibitem{jellyfish}
A.~Singla, C.-Y. Hong, L.~Popa, and P.~B. Godfrey, ``Jellyfish: Networking data
  centers randomly,'' in \emph{Proc. of USENIX NSDI}, 2012, pp. 225--238.

\bibitem{splaynet}
S.~Schmid, C.~Avin, C.~Scheideler, M.~Borokhovich, B.~Haeupler, and Z.~Lotker,
  ``Splaynet: Towards locally self-adjusting networks,'' \emph{IEEE/ACM
  Transactions on Networking (ToN)}, vol.~24, no.~3, pp. 1421--1433, 2016.

\bibitem{solstice}
H.~Liu, M.~K. Mukerjee, C.~Li, N.~Feltman, G.~Papen, S.~Savage, S.~Seshan,
  G.~M. Voelker, D.~G. Andersen, M.~Kaminsky \emph{et~al.}, ``Scheduling
  techniques for hybrid circuit/packet networks,'' in \emph{Proc. ACM CoNext},
  2015, pp. 1--13.

\bibitem{hanauer2022fast}
K.~Hanauer, M.~Henzinger, S.~Schmid, and J.~Trummer, ``Fast and heavy disjoint
  weighted matchings for demand-aware datacenter topologies,'' \emph{arXiv
  preprint arXiv:2201.06621}, 2022.

\end{thebibliography}

\end{document}